\RequirePackage{silence}

\documentclass[letterpaper,11pt]{article}
\usepackage[in]{fullpage}
\usepackage[normalem]{ulem}
\usepackage{amsmath}
\usepackage{amssymb}
\usepackage{amsthm}
\usepackage{amsfonts}
\usepackage{thmtools, thm-restate}
\usepackage{xcolor}
\usepackage{soul}
\usepackage{xfrac}
\usepackage{xspace}
\usepackage{comment}
\usepackage{embedfile}
\usepackage{algorithm}
\usepackage[noend]{algpseudocode}
\usepackage{parskip}
\usepackage{multirow}
\usepackage{calc}
\usepackage{hyperref}
\hypersetup{
	colorlinks,
	citecolor=blue,
	filecolor=black,
	linkcolor=black,
	urlcolor=black,
	linktoc=all
}
\newcounter{savealgorithm}
\newenvironment{subalgorithms}
{%
	\stepcounter{algorithm}%
	\edef\currentthealgorithm{\thealgorithm}%
	\setcounter{savealgorithm}{\value{algorithm}}%
	\setcounter{algorithm}{0}%
	\renewcommand{\thealgorithm}{\currentthealgorithm\alph{algorithm}}%
}
{%
	\setcounter{algorithm}{\value{savealgorithm}}%
}

\newtheorem{lemma}{Lemma}[section]
\newtheorem{theorem}[lemma]{Theorem}
\newtheorem{proposition}[lemma]{Proposition}
\newtheorem{corollary}[lemma]{Corollary}
\newtheorem{observation}[lemma]{Observation}

\newtheorem*{claim*}{Claim}
\newtheorem*{corollary*}{Corollary}
\newtheorem*{observation*}{Observation}

\theoremstyle{definition}\newtheorem{definition}[lemma]{Definition}
\theoremstyle{definition}\newtheorem{building-block}[lemma]{Building Block}

\newcommand{\accept}[0]{\textsc{accept}\xspace}
\newcommand{\reject}[0]{\textsc{reject}\xspace}

\newcommand{\cond}{\middle |}
\newcommand{\floor}[1]{{\left\lfloor{#1}\right\rfloor}}
\newcommand{\ceil}[1]{{\left\lceil{#1}\right\rceil}}

\newcommand{\Let}{\State \textbf{let} }

\newcommand{\E}{\mathop{{\rm E}\/}}
\newcommand{\Ct}{\mathop{{\rm Ct}\/}}
\newcommand{\eps}{\varepsilon}
\newcommand{\supp}{\mathrm{supp}}

\title{Support Testing in the Huge Object Model}
\author{Tomer Adar\thanks{Technion - Israel Institute of Technology, Israel. Email: \href{mailto:tomer-adar@campus.technion.ac.il}{tomer-adar@campus.technion.ac.il}.} \and Eldar Fischer\thanks{Technion - Israel Institute of Technology, Israel. Email: \href{mailto:eldar@cs.technion.ac.il}{eldar@cs.technion.ac.il}. Research supported by an Israel Science Foundation grant number 879/22.} \and Amit Levi\thanks{University of Haifa and Technion - Israel Institute of Technology, Israel. Email: \href{mailto:alevi@ds.haifa.ac.il}{alevi@ds.haifa.ac.il}.}}

\begin{document}
	\begin{titlepage}
		\maketitle
		\thispagestyle{empty}
		
		\begin{abstract}
			The Huge Object model is a distribution testing model in which we are  given access to independent samples from an unknown distribution over the set of strings $\{0,1\}^n$, but are only allowed to query a few bits from the samples.
			We investigate the problem of testing whether a distribution is supported on $m$ elements in this model. It turns out that the behavior of this property is surprisingly intricate, especially when also considering the question of adaptivity.
			
			We prove lower and upper bounds for both adaptive and non-adaptive algorithms in the one-sided and two-sided error regime. Our bounds are tight when $m$ is fixed to a constant (and the distance parameter $\eps$ is the only variable). For the general case, our bounds are at most $O(\log m)$ apart.
			In particular, our results show a surprising $O(\log \eps^{-1})$ gap between the number of queries required for non-adaptive testing as compared to adaptive testing. For one sided error testing, we also show that a $O(\log m)$ gap between the number of samples and the number of queries is necessary.
			Our results utilize a wide variety of combinatorial and probabilistic methods.
		\end{abstract}
		
	\end{titlepage}
	%%\tableofcontents
	%%\thispagestyle{empty}
	%\end{titlepage}
	%%\newpage
	\pagenumbering{arabic}
	
	\section{Introduction}
	Property testing~\cite{RS96,GGR98} is a framework concerned with analyzing global properties of an input while reading only a small part thereof, in the form of queries. Over the past few decades property testing has become an active field of study in theoretical computer science (see e.g,~\cite{G17}). The study of distribution property testing was first implicitly explored in \cite{gr11}, and explicitly formulated in \cite{batuFFKRW2001} and \cite{batu-FRSW2000}. In the standard model of distribution testing, an algorithm can access a sequence of independent sampled elements drawn from an unknown input distribution $\mu$, and it either accepts or rejects the input based on this sequence. An $\eps$-testing algorithm for a property of distributions is required to accept every input distribution that satisfies the property with high probability (e.g., $\frac{2}{3}$), and to reject with high probability (e.g., $\frac{2}{3}$) every input distribution whose variation distance from every distribution satisfying the property is greater than $\eps$.
	
	The standard model of distribution testing assumes that the elements drawn from the distribution are fully accessible, which might be unreasonable if they are ``huge''. The Huge Object model, whose study was initiated in \cite{gr22}, treats the sampled elements as long strings that have to be queried. In this model, for example, it is possible that the algorithm has two non-identical samples without being able to distinguish them. This ``two-phase'' characteristic of the Huge Object model (``sample then query'', rather than only taking samples or only querying a string) exhibits rich behavior with respect to adaptive querying, as studied in detail in \cite{adar23}.
	
	In the standard model of distribution testing, \cite{vv11} and \cite{vv17} show a tight bound of $\Theta(m / \log m)$ samples for two-sided error $\eps$-testing of having a support size bounded by $m$ in the standard model, for every fixed $\eps$. An upper bound of $O(\eps^{-1} m)$ samples for one-sided algorithms is implicitly shown in \cite{adar23}, and here we show that it is tight (Proposition \ref{prop:lbnd-supp-na-catch}). Based on these tight bounds, the bounded support property is considered to be fully understood in the standard model for one-sided testing, and mostly understood in the two-sided case (for every fixed $m$ there is still a gap between $\Omega(\eps^{-1})$  and $O(\eps^{-2})$ for two-sided testing).
	
	One would expect that having bounded support, which is arguably the simplest of distribution properties, would have simple and easily understood testing bounds also in the Huge Object model. As in the standard model, it is the only label-invariant property that is testable using one-sided error algorithms (Proposition \ref{prop:lbl-inv-and-one-side-is-supp}). However, it turns out that the behaviour of this property under the Huge Object model is surprisingly intricate. One unexpected feature that we show here is a gap between the number of queries required for non-adaptively testing for this property as compared to adaptive testing. Indeed there is no adaptivity in the standard distribution testing model, but one would not expect the label-invariant (and even mapping-invariant as per the definition in \cite{gr22}) property of having bounded support to exhibit such a gap.
	
	\subsection{Definition of the model}
	The Huge Object model differs from the standard sampling model in its distance measure and in the way that the algorithm gathers information about the input distribution.
	
	\subsubsection*{Algorithmic model}
	
	A probabilistic algorithm $\mathcal{A}$ with $q$ queries and $s$ samples, whose input is a distribution $P$ over $\{0,1\}^n$ accessible via the Huge Object model, is an algorithm that acts in the following manner: at every stage, the algorithm may ask for a new sample $v$ that is provided by drawing it according to $P$, independently of all prior samples, or may ask to query a coordinate $j\in\{1,\ldots,n\}$ of an old sample $u$ (the algorithm may use internal coin tosses to make its decisions). When this query is made, the algorithm is provided with $u_j\in\{0,1\}$ as its answer. The algorithm has no access to the sampled vectors apart from queries. At the end, after taking not more than a total of $s$ samples and making a total of not more than $q$ queries, the algorithm provides its output.
	
	We say that the algorithm is \emph{non-adaptive} if it makes all its sampling and querying decisions in advance, prior to receiving all query answers in bulk. Only the final output of a non-adaptive algorithm may depend on the received answers. The formal definition of adaptivity appears in Subsection \ref{sec:additional-prelims:subsec:algos}.
	
	\subsubsection*{Distances}
	
	Here we define some measures of distance. Note that we usually use $d(\cdot,\cdot)$ without mentioning the measure that we use, if its context is unambiguous. For distributions over $\{0,1\}^n$, $d(\cdot,\cdot)$ usually refers to the earth mover's distance.
	
	\begin{definition}[String distance]
		Let $u, v \in \{0,1\}^n$ be two strings. We define their \emph{distance} as the normalized Hamming distance,
		\[d_H(u,v) = \frac{1}{n}\left|\left\{ 1 \le i \le n \;\cond\; u_i \ne v_i \right\}\right| = \Pr_{i \sim \{1,\ldots,n\}}[u_i \ne v_i]\]
		We define the distance of $u\in\{0,1\}^n$ from a set $A\subseteq \{0,1\}^n$ as $d_H(u,A) = \min_{v \in A} d_H(u,v)$.
	\end{definition}
	
	\begin{definition}[Transfer distribution]
		Let $P$ and $Q$ be distributions over finite sets $\Omega_1$ and $\Omega_2$, respectively. A distribution $T$ over $\Omega_1 \times \Omega_2$ is a \emph{transfer distribution} from $P$ to $Q$ if for every $a \in \Omega_1$, $\Pr_{(u,v) \sim T}[u=a] = P(a)$, and for every $b \in \Omega_2$, $\Pr_{(u,v) \sim T}[v=b] = Q(b)$. The set of transfer distributions from $P$ to $Q$ is denoted by $\mathcal{T}(P,Q)$. Note that this is a compact set when considered as a set of real-valued vectors.
	\end{definition}
	
	\begin{definition}[Variation distance]
		Let $\mu$ and $\nu$ be two distributions over a finite set $\Omega$. Their \emph{variation distance} is defined as:
		\[d_\mathrm{var}(\mu,\nu)
		= \frac{1}{2}\sum_{u \in \Omega} \left|\mu(u) - \nu(u)\right|
		= \max_{E \subseteq \Omega} \left|\Pr_\mu\left[E\right] - \Pr_\nu\left[E\right]\right|
		= \min_{T \in \mathcal{T}(\mu,\nu)} \Pr_{(u,v) \sim T}\left[u \ne v\right] \]
	\end{definition}
	
	\begin{definition}[Earth mover's distance]
		Let $P$ and $Q$ be two distributions over $\{0,1\}^n$. Their \emph{earth mover's distance} is defined as:
		\[d_\mathrm{EMD}(P,Q) = \min_{T\in \mathcal{T}(P,Q)} \E_{(u,v) \sim T}\left[d_H(u,v)\right] \]
		The above minimum exists since it is in particular the minimum of a continuous function over a compact set.
	\end{definition}
	
	\subsubsection*{Testing model}
	\begin{definition}[A property]
		A \emph{property} $\mathcal{P}$ is a sequence $\mathcal{P}_1,\mathcal{P}_2,\ldots$ such that for every $n \ge 1$, $\mathcal{P}_n$ is a compact subset of the set of all distributions over $\{0,1\}^n$.
	\end{definition}
	
	\begin{definition}[Distance of a distribution from a property]
		Let $\mathcal{P} = (\mathcal{P}_1,\mathcal{P}_2,\ldots)$ be a property and $P$ be a distribution over $\{0,1\}^n$ for some $n$. The \emph{distance of $P$ from $\mathcal{P}$} is defined as $d_\mathrm{EMD}(P,\mathcal{P}) = \min_{Q \in \mathcal{P}_n}\{d_\mathrm{EMD}(P,Q)\}$.
	\end{definition}
	
	\begin{definition}[$\eps$-test]
		Let $\mathcal{P}$ be a property of distributions over $\{0,1\}^n$. We say that a probabilistic algorithm $\mathcal{A}$ is an \emph{$\eps$-test} for $\mathcal{P}$ if:
		\begin{itemize}
			\item For every $P \in \mathcal{P}$, $\mathcal{A}$ accepts with probability higher than $\frac{2}{3}$.
			\item For every probability distribution $P$ over $\{0,1\}^n$ that is $\eps$-far from $\mathcal{P}$ (satisfying $d(P,\mathcal{P})>\eps$), $\mathcal{A}$ rejects with probability higher than $\frac{2}{3}$
		\end{itemize}
	\end{definition}
	
	\begin{definition}[one-sided and two-sided $\eps$-test]
		Consider the setting of the above definition. If additionally for every input $P \in \mathcal{P}$, $\mathcal{A}$ accepts $\mathcal{P}$ with probability $1$ (rather than ``higher than $\frac{2}{3}$''), then we say that $\mathcal{A}$ is a \emph{one-sided $\eps$-test} for $\mathcal{P}$. Otherwise, we say that $\mathcal{A}$ has \emph{two-sided error}.
	\end{definition}
	
	%\ todo{\sout{(a) Exact definition of distances and testing (probabilities including 1-sided and distances but no algrotihmic structure) (b) Informal definition of adaptivity and say where it is defined formally}}
	
	\subsection{Summary of our results}
	
	\paragraph{Table of results}
	The following is a table summarizing the bounds presented here for $\eps$-testing for being supported by at most $m$ elements, along with previously known ones provided for reference (Section \ref{sec:quick-bounds} contains a sketch on how to derive them). The hidden coefficients in the $O(\cdot)$ and the $\Omega(\cdot)$ notations are global numerical constants. The new results appear in purple.
	
	\begin{center}
		\begin{tabular}{ |c|c|c| }
			\hline
			Model & One-sided Error & Two-sided Error \\
			\hline
			Standard model & $\Theta(\eps^{-1} m)$ & $\Omega(\eps^{-1} m / \log m)$ \cite{vv11} \\
			(Sample complexity) & Folklore, see \cite{adar23} & $O(\eps^{-2} m / \log m)$ \cite{vv17} \\
			\hline
			Huge Object & {\color{purple}$\Omega(\eps^{-1} m (\log \eps^{-1} + \log m))$} & {\color{purple} $\Omega(\eps^{-1} \log \eps^{-1}) $}\\
			Non-adaptive & {\color{purple}$O(\eps^{-1}m \log \eps^{-1} \log m)$} & $O(\eps^{-3} m \log \eps^{-1})$ \cite{vv17} + \cite{gr22} \\
			\hline
			Huge Object & {\color{purple}$\Omega(\eps^{-1} m \log m)$} & $\Omega(\eps^{-1} m / \log m)$  \cite{vv11} \\
			Adaptive & {\color{purple}$O(\eps^{-1}m \log m \cdot \min\{\log \eps^{-1}, \log m\})$} & \\
			\hline
		\end{tabular}
	\end{center}

	%\ todo{Change and \sout{add formal statements following the definition} do tools part}
	
	The overview in Section \ref{sec:overview} provides an informal guide on deriving our results, whose proofs appear in the sections following it.
	
	The following are some conclusions to be drawn from the bounds given above. We use $\mathcal{S}_m$ to denote the property of being supported by at most $m$ elements (formally defined in Definition \ref{def:Sm-SA-Sf}).
	
	\paragraph{Adaptive vs.\ non-adaptive two-sided asymptotic gap} The most surprising result is that testing a distribution for being supported by at most two elements cannot be done using a number of queries linear in $\eps^{-1}$, even with two-sided error. This result applies for every $m \ge 2$, and the exact bound is $\Omega(\eps^{-1} \log \eps^{-1})$ (with the implicit coefficient being independent of $m$). To the best of our knowledge, combined with the $O(\eps^{-1})$ adaptive upper bound of \cite{adar23}, ``being supported by at most two elements'' is the first explicit example of a property that is closed under mapping (and in particular is label-invariant) which has different asymptotic bounds for the number of queries for adaptive algorithms and non-adaptive ones in the Huge Object model (see Theorem \ref{th:m-supp-na-lbnd-two-sided}).
	
	A possible explanation for this is that being label-invariant in the Huge Object model is different from being so in the standard model, because applying a permutation on the labels may change their distinguishability, and in particular it may change the distance from the property.
	
	In this paper we provide a thorough investigation of $\mathcal{S}_m$ utilizing a variety of methods. In particular, we show several gaps such as the above mentioned one. However, the behaviour of the bounded support property in the Huge Object model, especially when considering it as a problem with two variables (namely the maximal support sized $m$ and the distance parameter $\eps$) is still not completely understood. We do have tight bounds for the fixed constant $m$ cases (where only $\eps$ is variable) for all algorithm types, and bounds up to logarithmic factors for the more general cases.
	
	\paragraph{One-sided bounds and a gap from the standard model}
	We have tight bounds for $\eps$-testing of $\mathcal{S}_m$ for every fixed $m$ (and variable $\eps$) for both non-adaptive algorithms and adaptive ones. These bounds are also tight for every fixed $\eps$ (and variable $m$). Additionally, our bounds show a gap between the standard model (considering sample complexity) and the Huge Object model (considering query complexity). Consider the bounded support property as a sequence of individual properties, where for every $m \ge 2$, the $m$-th property is $\mathcal{S}_m$. We show that, if we only allow one-sided error tests, there is an $O(\log m)$ gap between the standard model of distribution testing and the Huge Object model. In the standard model, there exists a one-sided test for $\mathcal{S}_m$ at the cost of $O(\eps^{-1}m)$ samples. In the Huge Object model, there is a lower bound of $\Omega(\eps^{-1} m \cdot \log m)$ many queries for every one-sided $\eps$-test, even if it is adaptive. Note that the gap is between the number of \emph{samples} in the standard model and the number of \emph{queries} in the Huge Object model, which is the natural measure of complexity in this model.
	
	\subsection*{New tools}
	
	\paragraph{A new algorithmic paradigm}
	For the adaptive one-sided upper bound, we define a standalone algorithmic primitive, the ``fishing expedition'' paradigm, that repeatedly executes a subroutine until it reaches a predefined goal or when it finds out that it is no longer cost-effective (even if it did not reach the goal). We believe that this primitive will also be useful in future endeavors.
	
	\paragraph{A hybrid probabilistic-extremal analysis}
	We define a concept of ``valid composition''. Loosely speaking, it is an ordered subset of samples that become closer to each other as the sequence progresses, but are still $\eps$-far from each other. We use a hybrid probabilistic-extremal argument to show that for an input distribution that is $\eps$-far from $m$-support, with high probability, there exists a valid composition with at least $m+1$-elements.
	
	The hybrid probabilistic-extremal argument works as follows: we define some rank of valid compositions. If for every individual valid composition with at most $m$ elements, there is a high probability that it is not maximal (according to the rank), then globally there is a high probability that none of them is maximal. Hence, with high probability, the maximally-ranked valid composition within our samples must have at least $m+1$ elements.
	
	%For the non-adaptive one-sided upper bound, we define ``valid composition'' as an ordered subset of samples that virtually go denser, while still $\Omega(\eps)$-far from each other. We use it since it is easy to distinguish all elements in a valid composition, even without locating its subset among the other samples. To show that a long valid composition exists, we use a probabilistic-extremal analysis: if the probability of an object to be maximal (with respect to some partial order) is too low unless it belongs to a ``good'' set (such as the set of all long valid compositions), then there is a high probability that a ``good'' object exist (since every non-empty, finite set must have a maximal element).
	
	%We show that it is relatively easy to find a witness against $m$-support if we are given a sample set that contains a valid composition of size $m+1$ or more. We define some ranking of valid compositions, and then use an extremal analysis to show that for input distributions that are $\eps$ far from $\mathcal{S}_m$, with high probability, no valid composition with at most $m$ elements can has maximal rank (unless the input distribution is supported on $m$ elements or less). This implies that there exists a valid composition with at least $m+1$ elements.
	
	\paragraph{A new use for an old combinatorial result}
	For the adaptive one-sided lower bound, we use an old combinatorial result, that a biclique covering of the $m$-clique must have at least $m \log_2 m$ vertices \cite{hansel64,ks67}, to show that the every witness against $m$-support is at least $m \log m$ bits long, which makes it a lower bound to the number of queries. To apply a multiplicative factor of $\eps^{-1}$, which is pretty easy for non-adaptive algorithms, we analyze the effectivity of a decision tree that incrementally constructs a witness based on the queries.
	
	% \ todo{\sout{compositions (get inspiration from this year's Prob Methods question to explain hybrid argument), the use of contradiction graph (with in-decision-tree analysis)}}
	
	\subsection{Open problems}
	\paragraph{One-sided non-adaptive bounds} We have an $\Omega(\eps^{-1} m (\log \eps^{-1} + \log m))$ lower bound for one-sided $\eps$-testing of $\mathcal{S}_m$, as well as an $O(\eps^{-1} m \log \eps^{-1} \log m)$ upper bound for one-sided $\eps$-testing of $\mathcal{S}_m$. We believe that the upper bound is tight, but we do not have the corresponding lower bound. What is the true complexity of $\eps$-testing $\mathcal{S}_m$?
	
	\paragraph{Non-trivial two-sided bounds} Is there a lower bound of $\omega(m / \log m)$ queries for two-sided testing of $\mathcal{S}_m$ (noting that \cite{vv11} only gives $\Omega(m/\log m)$), even for non-adaptive algorithms? We believe that $\Omega(m)$ should be this lower bound, based on the $\log m$-gap in the one-sided case (a $\Theta(m)$ tight bound in the standard model, and a $\Theta(m \log m)$ tight bound in the Huge Object model).
	
	\paragraph{One-sided adaptive bounds} Our results for one-sided adaptive $\eps$-testing of $\mathcal{S}_m$ are tight with respect to $m$, but have a logarithmic gap with respect to $\eps$ (more precisely, with respect to $\min\{\eps^{-1}, m\}$). Closing this gap is an open problem.
	
	\paragraph{The tradeoffs between sample and query complexity} Our bounds apply to the query complexity of the tests. The lower bounds adapted from previous works on the traditional model clearly apply for the sample complexity here, even if we allow a higher query complexity. As for our new upper bounds, most of them have a polylogarithmic average queries per sample ratio. It would be interesting to investigate whether the sample complexity can be reduced if we allow a much higher (but still sub-linear in $n$) number of queries per sample.
	
	\section{Preliminaries}
	\label{sec:prelims}
	
	\subsection{Algorithmic model}
	
	As observed by Yao \cite{yao77}, every probabilistic algorithm can be seen as a distribution over a set of deterministic algorithms. Hence we can analyze probabilistic query-making algorithms by analyzing the deterministic algorithms they are supported on.
	
	We observe that we can assume that all samples are drawn before the first query is made, since they are fully independent: the distribution of every sample made does not depend at all on any calculation or queries that occurred before it was taken, and so we can assume that it was taken before any calculation was performed. Based on this observation we can represent our algorithms using a $\{0,1\}$-valued matrix (whose rows are sampled from the distribution), from which the algorithms are allowed to query. %a run of an algorithm as a restriction of a $0-1$-matrix whose rows are the samples.
	
	\begin{definition}[Matrix representation of input access]
		Considering an algorithm with $s$ samples and $q$ queries, we assume that the samples are all taken at the beginning of the algorithm and are used to populate a matrix $M\in\{0,1\}^{s\times n}$. Then, during the run of the algorithm, each of its queries is represented as a pair $(i,j)\in \{1,\ldots,s\}\times \{1,\ldots,n\}$, for which the answer is $M_{i,j}$.
	\end{definition}
	
	\begin{definition}[Adaptive algorithm]
		Every deterministic algorithm in the Huge Object model with $q$ queries over $s$ samples is equivalent to a pair $(T,A)$, where $T$ is a decision tree of height $q$ in which every internal node contains a query $(i,j)$ (where $1 \le i \le s$ is the index of a sample and $1 \le j \le n$ is the index to query), and $A$ is the set of accepting leaves.
	\end{definition}
	
	\begin{definition}[Non-adaptive algorithm]
		A deterministic algorithm $(T,A)$ with $q$ queries is \emph{non-adaptive} if, for every $0 \le i < q$, all internal nodes at the $i$-th level consist of the exact same query. Every non-adaptive algorithm can be represented as a pair $(Q,A)$, where $Q \subseteq \{1,\ldots,s\} \times \{1,\ldots,n\}$ is a \emph{set} of queries and $A \subseteq \{Q \mapsto \{0,1\} \}$ is the set of accepted answer vectors.
	\end{definition}
	
	\subsection{Technical components}
	
	\subsubsection*{Fishing expedition}
	We define an algorithmic primitive that allows us to repeat an execution of a probabilistic subroutine until it is no longer effective. Consider for example a ``coupon-collector'' type process, but one in which the number of distinct elements is not known to us. The goal is to collect a preset number of elements, but we also want to stop early if we believe that there are no more elements to be effectively collected. %The following two definitions formalize such processes. 
	%The exact formal definition appears in Section \ref{sec:fishing-expedition}. Here we describe it informally.
	
	Consider a (probabilistic) subroutine $\mathcal{A}$ that can either fail or succeed. We denote the outcome of an execution of $\mathcal{A}$ by $R$. In this discussion the outcome includes both the explicit output of the execution and its side effects, which may affect the probabilities for future executions of $\mathcal{A}$. We thus analyze a \emph{sequence} of executions $R_1,\ldots,R_N$, where $R_1$ is performed over the initial state. We define two behaviors of ``coupon collection'' that such an $\mathcal{A}$ must present.
	
	\begin{definition}[Fail stability] \label{def:fail-stability}
		Let $\mathcal{A}$ be a subroutine that may succeed or fail. Specifically let $R_1,\ldots,R_N$ be random variables that detail the outputs of the first $N$ executions of $\mathcal{A}$. %, where a $0$ value means that an execution has failed (and all other values imply success).
		We say that $\mathcal{A}$ is \emph{fail stable} with respect to a set $G$ of outcomes indicating success, if for every $2 \le i \le N$ and every result sequence $(r_1,\ldots,r_{i-1})\in\supp(R_1,\ldots,R_{i-1})$ for which $r_{i-1} \notin G$:
		\[\Pr\left[R_i \in G \mid R_1 = r_1, \ldots, R_{i-2} = r_{i-2}, R_{i-1} = r_{i-1} \right] = \Pr\left[R_{i-1} \in G \mid R_1 = r_1, \ldots, R_{i-2} = r_{i-2} \right]\]
		In other words, a failure does not affect the probability of further executions to succeed.
	\end{definition}
	
	\begin{definition}[Diminishing returns] \label{def:diminishing-returns}
		Let $\mathcal{A}$ and $R_1,\ldots,R_N$ be as in Definition \ref{def:fail-stability}. We say that $\mathcal{A}$ \emph{has diminishing returns} with respect to a set $G$ of successful outcomes, if for every $2 \le i \le N$ and every result sequence $(r_1,\ldots,r_{i-1})\in\supp(R_1,\ldots,R_{i-1})$:
		\[\Pr\left[R_i \in G \mid R_1 = r_1, \ldots, R_{i-2} = r_{i-2}, R_{i-1} = r_{i-1} \right] \le \Pr\left[R_{i-1} \in G \mid R_1 = r_1, \ldots, R_{i-2} = r_{i-2}\right]\]
		That is, if $\mathcal{A}$ has diminishing returns, then a success in a single execution never increases, but may decrease, the probability of further executions to succeed.
	\end{definition}
	
	Recall the coupon-collecting example. We expect it to have both fail stability and diminishing returns (with respect to a common set $G$ of outcomes indicating success). If we look for a coupon and do not find it in a single try, nothing happens. Further tries will have the same probability to succeed. On the other hand, if we collect a coupon, then in further tries, there are less uncollected coupons left and it is slightly harder to find an additional one.
	
	The fishing expedition paradigm seeks to collect a goal of $k$ coupons, but ``gives up'' if it believes that the probability to find an additional coupon is less than some parameter $p$. %Below is an informal version of Lemma \ref{lemma:fishing-expedition-full}.
	
	%\begin{lemma*}[The Fishing Expedition paradigm (informal statement)]
	%    Let $\mathcal{A}$ be a subroutine that features the behaviors of fail stability and diminishing returns.
	%    For every $p > 0$, $q > 0$, $k \ge 1$, there exists an algorithm that executes $\mathcal{A}$ at most $O(p^{-1}(H + \log q^{-1} + \log \log k))$ times, where $H$ is the number of successful executions. Moreover, with probability at least $1 - q$, the algorithm either have $k$ successful executions ($H=k$, reached our goal) or correctly reports that the success probability of an additional execution is at most $p$ (``coupons become too rare'').
	%\end{lemma*}
	
	The desired algorithm has three parameters: a threshold $p$, a confidence $q$ and a goal $k \ge 1$. The input is a subroutine $\mathcal{A}$ with diminishing returns and fail stability (with respect to some common set $G$). Informally, the goal of the algorithm is to have $k$ successful executions of $\mathcal{A}$, but also to terminate earlier if the probability of $\mathcal{A}$ to succeed becomes lower than $p$. Since the algorithm has no actual access to the success probability of $\mathcal{A}$, it should terminate early only if it is confident enough that the success probability of further executions is too low for them to be effective.
	
	\begin{lemma}
		\label{lemma:fishing-expedition-full}
		Consider a black box subroutine $\mathcal{A}$ with fail stability (Definition \ref{def:fail-stability}) and diminishing returns (Definition \ref{def:diminishing-returns}) with respect to a common set $G$ of outcomes indicating success.
		
		For an algorithm that repeatedly executes $\mathcal{A}$, we define the following random variables:
		\begin{itemize}
			\item $N$ -- the number of executions.
			\item $R_1,\ldots,R_N$ -- their outcomes.
			\item $X_1,\ldots,X_N$ -- indicators of success (that is, $X_i = 1$ if and only if $R_i \in G$).
			\item $H = \sum_{i=1}^N X_i$ -- the number of successful executions.
			\item $\hat p = \Pr[X_{N+1} = 1 | R_1,\ldots,R_N]$ -- the success probability of a possible extra execution of $\mathcal{A}$.
		\end{itemize}
		
		Considering the parameters $p > 0$ (threshold), $q > 0$ (confidence), and $k \ge 1$ (goal), there exists an algorithm that repeatedly executes $\mathcal{A}$ for which $N \le p^{-1}(4H + 5(\log q^{-1} + \log (\log k + 1))) + 1$ and $H \le k$, such that with probability higher than $1-q$, either $H = k$ or $\hat p \le p$ (or both).
	\end{lemma}
	
	\subsubsection*{Contradiction graph}
	We define here what it means to be a ``counter-example'' for having a bounded support size $m$.
	
	\begin{definition}[Contradiction graph] \label{def:contradiction-graph}
		Let $x_1,\ldots,x_s \in \{0,1\}^n$ be a sequence of strings. Let $Q \subseteq \{1\ldots,s\} \times \{1,\ldots,n\}$ be a set of queries. We define \emph{the contradiction graph} of $(x_1,\ldots,x_s;Q)$ as $G(V,E)$ with $V=\{1,\ldots,s\}$, and for every $1 \le i_1,i_2 \le s$:
		\[\{i_1,i_2\} \in E \iff \exists 1 \le j \le n : (x_{i_1})_j \ne (x_{i_2})_j \wedge \left((i_1,j),(i_2,j) \in Q\right)\]
		Note that the graph is undirected since the definition of the edges is commutative. It is also clearly without self-loops.
	\end{definition}
	
	\begin{definition}[Witness against $m$-support] \label{def:witness-against-Sm-not-colorable}
		Let $P$ be a distribution that is supported by a set of more than $m$ elements. We say that $(x_1,\ldots,x_s; Q)$ is a \emph{witness against $m$-support} (of $P$) if $x_1,\ldots,x_s$ are all drawn from $P$, and their contradiction graph is not $m$-colorable.
	\end{definition}
	
	We prove in Lemma \ref{lemma:witness-iff-not-colorable} that calling the above a witness is indeed justified, in the sense that a distribution $P$ has $m$-support if and only if there is zero probability to draw a tuple $x_1,\ldots,x_s$ for which one can provide a query set $Q$ that makes it a witness.
	
	\begin{definition}[Explicit witness against $m$-support]
		Let $P$ be a distribution that is supported by a set of more than $m$ elements. We say that $(x_1,\ldots,x_s,Q)$ is an \emph{explicit witness against $m$-support} (of $P$) if $x_1,\ldots,x_s$ are all drawn from $P$, and their contradiction graph contains a clique with $m+1$ vertices as a subgraph.
	\end{definition}
	Note that an explicit witness is in particular a witness against $m$-support, but the converse does not generally hold.%in particular an explicit witness is a witness against $m$-support.

	\section{Overview of our proofs}
	\label{sec:overview}
	
	%The following is a guide of the new ideas and proofs of our results, and their location in the paper. Section \ref{sec:quick-bounds} provides a sketch on how to quickly derive the bounds that we use from previous works.
	
	\subsection*{Two-sided, non-adaptive lower-bound (Theorem \ref{th:m-supp-na-lbnd-two-sided})}
	
	%\ todo{Do over after other todos\ldots this particular part is actually Ok-ish}
	
	We first describe our lower bound for $\mathcal{S}_2$, which holds the main ideas also for $\mathcal{S}_m$. We begin by analyzing a restricted form of non-adaptive algorithms, which we call \emph{rectangle algorithms}. A rectangle algorithm is characterized by the number of samples $s$ and a set $I$ of indices. Every sample is queried at the indices of $I$, hence the query complexity is $s \cdot |I|$. We say that $|I|$ is the ``width'' of the rectangle and that the number of samples is its ``height''.
	
	Consider the following $O(\eps^{-1})$-query rectangle algorithm: for some hard-coded parameter $\beta > 0$, it chooses a set $I$ of $O(\beta^{-1})$ indices, and then it takes $O(\beta \eps^{-1})$ samples, and then queries every sample on all indices of $I$.
	
	Now consider the following form of inputs. For some $\alpha > 0$ and two strings $a$ and $b$ for which $d(0,a), d(0,b), d(a,b) = \Theta(\alpha)$, let $P$ be the following distribution. The string $0$ is picked with probability $1 - c\alpha^{-1} \eps$, the string $a$ with probability $\frac{c}{2}\cdot\alpha^{-1} \eps$ and the string $b$ with probability $\frac{c}{2}\cdot\alpha^{-1} \eps$, where $c>1$ is some global constant.
	
	Intuitively, the algorithm finds a witness against $2$-support if there is a query common to $a$ and $b$, at an index $j$ that is not always zero (we call such $j$ a \emph{non-zero index}). That is, there are two necessary conditions to reject: the algorithm must get both $a$ and $b$ as samples, and it must query at an index $j$ for which $(a)_j \ne (b)_j$.
	
	The expected number of non-zero samples that the algorithm gets is $O(\alpha^{-1}\beta)$. If $\alpha$ is much greater than $\beta$, then with high probability the algorithm only gets all-zero samples and cannot even distinguish the input distribution from the deterministic all-zero one.
	
	The expected number of non-zero indices that the algorithm chooses is $O(\alpha \beta^{-1})$. If $\alpha$ is much smaller than $\beta$, then with high probability all queries are made in ``zero indices'' and the algorithm again cannot even distinguish the input distribution from the deterministic all-zero one. Thus, the algorithm can reject the input with high probability only if $\alpha \approx \beta$.
	
	Our construction of $D_\mathrm{no}$ chooses $\alpha=2^k$ where $k$ is distributed uniformly over its relevant range, to ensure that a rectangle algorithm (with a fixed $\beta$) ``misses'' $\alpha$ with high probability. Intuitively, the idea is that a non-adaptive algorithm must accommodate a large portion of the possible values of $\alpha$, which would lead to an additional $\log \eps^{-1}$ factor. Then, we show that given an input drawn from $D_\mathrm{no}$, if the algorithm did not distinguish two non-zero elements, then the distribution of runs looks exactly the same as the distribution of runs of the same algorithm given an input drawn from $D_{\mathrm{yes}}$, which is supported over $0$ and a single $a$.
	
	To show that the above distributions defeat any non-adaptive algorithm (not just rectangle algorithms), we analyze every index $1\leq j\leq n$ according to the number of samples which are queried in that index. If few samples are queried, then this index has a high probability of not hitting two non-zero samples, rendering it useless (we gain an important advantage by noting that actually querying $j$ from at least two non-zero samples is required for it to be useful). If many samples are queried then this index may hit many samples, but only few indices can host many queries, which gives us a high probability of all of them together not containing a non-zero index among them.
	
	To extend this result to $m\ge 2$, for every $t \ge 2$ we define a distribution $D^{t}_\mathrm{no}$ over inputs that are supported by $t+1$ elements (one of them being the zero vector), and also $\eps$-far from being supported by $m$ elements (for every $m \le t/2 + 1$). As before, we define $D_\mathrm{yes}$ as a distribution over inputs supported by $2$ elements, which is identical to $D^1_\mathrm{no}$ (including the all-zero elements), and then we proceed with the same argument as before.

	\subsection*{One-sided, non-adaptive upper bound (Theorem \ref{th:alg-m-supp-na-correct})}
	
	%\ todo{Rewrite, give more information on what is a composition, together with new ``new tools'' in the intro give an idea that we prove that a large one exists (note synergy with new intro portion)}
	
	Let us first consider a ``reverse engineering'' algorithm: for every $\ell = 2^0, 2^1, \ldots, 2^{\log \eps^{-1}}$, we query $\Theta((\eps^{-1} / \ell) \cdot \log m)$ indices that are common to at least $\ell \cdot m$ samples. Intuitively, according to the analysis of the two-sided lower bound, the algorithm should have roughly $\Omega(m \log m)$ indices that distinguish pairs of elements, which suffice for a contradiction graph that contains an $m+1$-clique.
	
	This intuition appears to be lacking when it comes to showing the correctness of this construction for inputs that lack the special form of $D^t_\mathrm{no}$. To be able to handle distance combinations (instead of just one ``$\alpha$'' as above), we use a concept of ``valid compositions''. A valid composition is an ordered combination of samples ($x_1,\ldots,x_k$) and a sequence of non-decreasing scales ($a_2,\ldots,a_k$), for which the distances are bounded by $d(x_i, \{x_1,\ldots,x_{i-1}\}) > 2^{-a_i-1}$ (see Definitions \ref{def:composition-soundness}, \ref{def:composition-monotonity}, \ref{def:composition-valid}).
	
	Querying according to index sets whose random choice follows the prescribed distances distinguishes all elements in a composition with high probability. Our goal is to show the existence of valid compositions of $m+1$ elements in order to ensure that we find an explicit witness, and thus establish the upper bound. However, it is not clear that ``long'' valid compositions even exist.
	
	We use an extremal probabilistic argument, and show that if the input is $\eps$-far from having support size $m$, then with high probability no ranked composition with at most $m$ elements is maximal. This implies that the maximally ranked composition (with respect to an order that we define) cannot have less than $m+1$ elements, leading with high probability to finding an explicit witness against $m$-support. %unless the input is indeed $\eps$-close to having support size $m$.

	\subsection*{One-sided, adaptive upper bound (Theorem \ref{th:alg-m-supp-adaptive-correct})}
	We adaptively construct a distinguishing sequence that resembles a valid composition, but at some point we decide to ``give up'' and change phase to another way of querying that is more efficient under some conditions. Luckily, the condition that makes us give up implies them.
	
	For every distance scale, from $\Omega(1)$ to $\frac{1}{m}$, we use the ``fishing expedition'' paradigm to extend our sequence with as many elements as we can until we are certain enough that it is no longer effective to look for them (or until we find a witness against $m$-support). Unfortunately, it is possible that at some point the algorithm is certain enough that it is no longer effective to look for elements in any of these scales. At this point, we observe that the contribution of elements with small distance scale to the distance of the input from $\mathcal{S}_m$ is still $\Omega(\eps)$ (that is, we can safely ignore the ``rare large-distance elements''). To make use of this observation, the algorithm shifts to the second phase of looking for elements with small distances in a more general way which does not necessarily follow the ``theme'' of looking for valid compositions.
	
	In the small distance scale phase we construct and maintain a ``decision tree'' data structure over the existing elements, so that for every element that we need to compare to the existing elements, we can rule out in advance, using only $O(m)$ many queries, all but one of them. This allows us to save queries, since the smaller distances require the querying of relatively many indices for a comparison, which would have been very inefficient to perform for all existing elements.

	\subsection*{One-sided lower-bounds (Theorem \ref{th:m-supp-na-one-sided-lbnd-composite}, Corollary \ref{col:lbnd-one-sided-m-adaptive})}
	We prove that an algorithm obtains a witness against $m$-support if and only if the contradiction graph (Definition \ref{def:contradiction-graph}) is not $m$-colorable (Lemma \ref{lemma:witness-iff-not-colorable}). Hence we look for the lower bound on the number of queries needed to construct a non-$m$-colorable contradiction graph.
	
	We observe that, given a query set, every index $j$ describes a biclique contradiction graph whose classes are ``all samples queried at $j$ for which $x_j=0$'' and ``all samples queried at $j$ for which $x_j=1$''. The contradiction graph is the union of these graphs. Then we extend our analysis in two ways, one of which applies to non-adaptive algorithms (giving a $\log \eps^{-1}$ factor) and the other also applies to adaptive ones (giving a $\log m$ factor).
	
	For non-adaptive algorithms, we extend the analysis of the two-sided bound to show that a one-sided algorithm for $\mathcal{S}_m$ requires $\Omega(\eps^{-1} m \log \eps^{-1})$ many queries. The following shows the hardness of ``gathering a witness against $\mathcal{S}_m$'', which allows for a more versatile argument as compared to the indistinguishability argument that we use for the lower bound of Theorem \ref{th:m-supp-na-lbnd-two-sided}.
	
	We use $D^t_\mathrm{no}$ using $t=4m/3$. For a non-adaptive algorithm that makes less than $O(\eps^{-1} m \log \eps^{-1})$ queries, the probability that it distinguishes two specific non-zero elements is $\frac{1}{16}$. Considering the contradiction graph, excluding the vertex corresponding to the zero vector, we show that the expected number of edges is at most $\frac{1}{16}\binom{t}{2}$. By Markov's inequality, with probability higher than $\frac{2}{3}$, there are less than $\binom{3t/4 - 1}{2} = \binom{m-1}{2}$ edges, meaning that this subgraph is colorable using $m-1$ colors. Combined with the vertex corresponding to the zero vector, the contradiction graph is colorable by $m$ colors, hence it cannot be a witness against being supported on only $m$-support.
	
	For the other bounds we extend a result of \cite{hansel64, ks67}, providing a lower bound on bi-clique covers of cliques, to show that every biclique cover of a non-$m$-colorable graph requires at least $(m+1) \log_2 (m+1)$ vertices (Lemma \ref{lemma:clique-witness-reduction}). To show another lower bound against non-adaptive algorithms, we construct a distribution in which a single, ``anchor'' element is drawn with probability $1 - \Theta(\eps)$. This way, for every non-adaptive algorithm that makes only $o(\eps^{-1} m \log m)$ many queries, the expected number of queries applied to other elements is $o(m \log m)$. By Markov's inequality, with probability $\frac{2}{3}$, only $o(m \log m)$ queries are made in non-zero elements, and in this case, there cannot be a witness against $m-1$ other elements.
	
	This construction cannot immediately be applied to adaptive algorithms, since they can use adaptivity to avoid wasting queries on the anchor element. To overcome this issue, we use two additional methods. The first one is using very short strings, that is, we focus on distributions over $\{0,1\}^{O(\log m)}$ that are $\eps$-far from having $m$ elements in their support (later we prove that the bound also holds for arbitrarily large $n$ using a simple repetition technique). The second method involves using shared-secret code ensembles \cite{beflr2020} that guarantee, in an appropriate setting, that if the algorithm makes less than $O(\log m)$ queries in an individual sample, then it gathers no information at all. This way, for every individual sample, the algorithm either behaves similarly to a non-adaptive algorithm or makes at least a fixed portion of the maximum number of queries. The exact argument requires a careful analysis of the decision tree of the algorithm.
	
	\section{Quick bounds from previous results}
	\label{sec:quick-bounds}
	We recall some known results for the standard model and use them to derive initial bounds on testing $\mathcal{S}_m$.
	
	%\ todo{\sout{check if this requires ``additional preliminaries'', if not then move to before it}}
	
	Observe that, without loss of generality, we can assume that every sample is queried at least once. Using distributions over sets of of vectors that are mutually $0.499$-far, lower bounds for the standard model can be converted to to the Huge Object model, implying in particular the following.
	
	\begin{proposition}[Proposition 2.8 in \cite{gr22}] \label{prop:vv-lbnd}
		Every two-sided error $\eps$-test for $\mathcal{S}_m$ makes at least $\Omega(m / \log m)$ queries (for some fixed $\eps$).
	\end{proposition}
	
	In the Huge Object model, different samples may be indistinguishable, hence standard-model algorithms cannot be immediately converted to Huge Object model ones. However, we can use the following reduction.
	\begin{lemma}[Theorem 2.2 in \cite{gr22}] \label{lemma:gr-map-inv-reduction}
		Suppose that $\mathcal{P}$ is testable with sample complexity $s(n,\eps)$ in the standard model, and that $\mathcal{P}$ is closed under mapping (note that bounded support properties are closed under mapping). Then for every $\eps > 0$ there exists a non-adaptive $\eps$-test for $\mathcal{P}$ in the Huge Object model that uses $3\cdot s(m,\eps)$ samples and $O(\eps^{-1} \log (\eps^{-1} s(m,\eps/2)))$ queries per sample.
	\end{lemma}
	
	\begin{proposition}[combining \cite{vv17} and \cite{gr22}] \label{prop:vv-ubnd}
		There exists a two-sided $\eps$-test for $\mathcal{S}_m$ whose query complexity is $O(\eps^{-3} m \log \eps^{-1})$.
	\end{proposition}
	
	\begin{proof}
		In \cite{vv17} there is a two-sided algorithm for $m$-support testing that uses $O(\eps^{-2} m/\log m)$ samples in the standard model of distribution testing. Lemma \ref{lemma:gr-map-inv-reduction} implies that there exists a non-adaptive $\eps$-testing algorithm for $m$-support that uses $O(\eps^{-1} m \log \eps^{-1})$ queries per sample, which gives $O(\eps^{-3} m \log \eps^{-1})$ queries in total.
	\end{proof}
	
	In the above we used \cite{vv17} rather than the more recent \cite{wu2019chebyshev}, since we needed a statement that holds for all values of $\eps$ (including those smaller than $1/m$). Proposition \ref{prop:vv-ubnd} implies that for every fixed $\eps$ and variable $m$, there exists an $O(m)$ non-adaptive two-sided error $\eps$-test for $\mathcal{S}_m$. In this context we also note the following known bounds.
	
	\begin{theorem}[\cite{gr22}, Corollary 2.3]
		For every $\eps > 0$ and $m \ge 2$, there exists a non-adaptive  one-sided $\eps$-testing algorithm for $\mathcal{S}_m$ that takes $O(\eps^{-1} m)$ samples and makes $O(\eps^{-2} m \log (m/\eps))$ queries.
	\end{theorem}
	\begin{theorem}[\cite{adar23}, Theorem 6.1]
		For every $\eps > 0$ and $m \ge 2$, there exists an adaptive  one-sided $\eps$-testing algorithm for $\mathcal{S}_m$ that takes $O(\eps^{-1} m)$ samples and makes $O(\eps^{-1} m^2)$ queries.
	\end{theorem}
	This immediately implies an upper bound of $O(\eps^{-1}m)$ samples for $\eps$-testing $\mathcal{S}_m$ in the standard model of distribution testing. As can be expected, this is tight. The following proposition is considered common knowledge, but as we are not aware of any reference proof, we put one here.
	
	\begin{proposition}
		\label{prop:lbnd-supp-na-catch}
		Every one-sided $\eps$-test for $\mathcal{S}_m$ takes at least $\Omega(\eps^{-1} m)$ samples in the standard model.
	\end{proposition}
	\begin{proof}
		For $\eps < \frac{1}{24}$, let
		\begin{align*}
			\mu : \begin{cases}
				1 & \mathrm{with\ probability\ }  1 - 2\eps \\
				i & \mathrm{with\ probability\ } \frac{1}{m} \eps, 2 \le i \le 2m+1
			\end{cases}
		\end{align*}
		
		The variation distance of $\mu$ from $\mathcal{S}_m$ is greater than $\eps$, since for every set of $m$ elements, there are additional $m+1$ elements in the support of $\mu$ whose combined probability is at least $\frac{m+1}{m}\eps$.
		
		Assume that we draw infinitely many independent samples $x_1,x_2,\ldots$. Let $\mathcal B$ be the event for $x_1 = 1$. For every $1 \le k \le m+1$, let $T_k$ be the index of the first $k$-th distinct element. Conditioned on $\mathcal B$, for every $2 \le k \le m+1$, $T_k - T_{k-1}$ distributes as a geometric variable with success probability $(2 - \frac{k-1}{2m})\eps \ge \eps$, hence its expected value is at least $\eps^{-1}$, and its variance is at most $\frac{1-\eps}{\eps^2}$. By linearity of expectation, 
		\[\E[T_{m+1} - T_1] \ge \Pr[\mathcal{B}]\E[T_{m+1} - T_1 | \mathcal{B}] = \Pr[\mathcal{B}]\sum_{k=2}^{m+1}\E[T_k - T_{k-1} | \mathcal B] \ge m\eps^{-1}\]
		The differences $T_k - T_{k-1}$ are independent, even if conditioned on $\mathcal B$, hence $\mathrm{Var}[T_{m+1} - T_1 | \mathcal B] = \sum_{k=2}^{m+1} \mathrm{Var}[T_k - T_{k-1} | \mathcal B] \le \frac{(1-\eps)m}{\eps^2}$.
		
		For $m \ge 16$, $\sqrt{\mathrm{Var}[T_{m+1} - T_1 | \mathcal B]} \le \eps^{-1}\sqrt{m} \le  \frac{1}{4} \eps^{-1}m \le \frac{1}{4} \E[T_{m+1} - T_1 | \mathcal B]$. By Chebyshev's inequality ($\Pr[X < \E[X] - \lambda \sqrt{\mathrm{Var}[X]}] < \lambda^{-2}$), we obtain
		\[\Pr[T_{m+1} - T_1 < \frac{1}{2}\eps^{-1}m | \mathcal B]  \le \Pr[T_{m+1} - T_1 < \E[T_{m+1} - T_1 | \mathcal B] - 2\sqrt{\mathrm{Var}[T_{m+1} - T_1 | \mathcal B]}] < \frac{1}{4}\]
		
		Hence, for an algorithm that takes $\floor{\frac{1}{2}\eps^{-1}m}$ samples, the probability to draw $m+1$ distinct samples is less than $\Pr[\neg \mathcal B] + \frac{1}{4} \le 2\eps + \frac{1}{4} < \frac{1}{3}$. Since a one-sided algorithm cannot reject without observing more than $m$ distinct members in the support of $\mu$, this concludes the proof.
	\end{proof}
	
	As with Proposition \ref{prop:vv-lbnd}, the above can be immediately converted to a Huge Object model bound.
	\begin{proposition}
		Every one-sided $\eps$-test for $\mathcal{S}_m$ in the Huge Object model must make at least $\Omega(\eps^{-1} m)$ queries as well.
	\end{proposition}
	In this paper we improve this proposition, showing a gap between the standard model and the Huge Object model for one-sided error tests.
	\section{Additional preliminaries} \label{sec:additional-prelims}
	
	% basic notations
	\subsection{Common notations}
	
	For an integer $n$ and a set $A \subseteq \{1,\ldots,n\}$, we denote by $1_A$ the $n$-bit binary string for which the $i$-th bit is $1$ if and only if $i \in A$. Given two sets $A, B \subseteq\{1,\ldots,n\}$, we let $A\Delta B$ be their symmetric difference. For a finite set $\Omega$, we define $\mathcal{D}(\Omega)$ as the set of all distributions over $\Omega$.
	
	The following is a useful notation for analyzing expectations of random variables.
	\begin{definition}[Contribution of a random variable over an event]
		Let $\mu$ be a probability distribution, $X$ be a random variable and $B$ be an event. We define the \em{contribution $\mathrm{Ct}[X\mid B]$ of $X$ over $B$} to be $0$ if $\mathrm{Pr}[B]=0$, and otherwise by
		\[\Ct\left[X\mid B\right]=\E[X\mid B]\cdot \Pr[B]=\sum_{x\in B}\mu(x)X(x).\]
	\end{definition}
	
	\begin{observation}
		The following properties are immediate.
		\begin{itemize}
			\item Inclusion-Exclusion: $\Ct\left[X \cond B_1 \vee B_2\right] = \Ct\left[X \cond B_1\right] + \Ct\left[X \cond B_2\right] - \Ct\left[X \cond B_1 \wedge B_2\right]$.
			\item Total expectation: If $\Pr\left[\bigvee_{i=1}^k B_i\right] = 1$ and the events $B_1,\ldots,B_k$ are mutually disjoint, then $\E\left[X\right] = \sum_{i=1}^k \Ct\left[X \cond B_i\right]$.
		\end{itemize}
	\end{observation}
	
	\begin{definition}[Sample map]
		Let $P$ be a distribution over a finite set $\Omega_1$ and let $f : \Omega_1 \to \Omega_2$ be a mapping to a finite set $\Omega_2$ (possibly $\Omega_1=\Omega_2$). The \emph{sample map} of $P$ according to $f$, denoted by $f(P)$, is the distribution $Q$ over $\Omega_2$ for which, for every $b \in \Omega_2$, $\Pr_Q[b] = \Pr_{a \sim P}[f(a) = b]$.
	\end{definition}
	
	% the bounded support property
	\begin{definition}[The bounded support property] \label{def:Sm-SA-Sf} We define the following variants of bounded support properties:
		\begin{enumerate}
			\item Let $m$ be a fixed number. The property of distributions that are supported by at most $m$ elements is denoted by $\mathcal{S}_m$.
			\item Let $A$ be a fixed set of elements. The property of distributions that are supported by a subset of $A$ (possibly $A$ itself) is denoted by $\mathcal{S}_A$.
			\item Let $f : \mathbb N \to \mathbb N$ be a fixed function. The property of distributions over $\{0,1\}^n$ that are supported by at most $f(n)$ elements is denoted by $\mathcal{S}_{f}$.
		\end{enumerate}
	\end{definition}

	% \subsection{Distances}
	%\ todo{to be mostly moved to intro} %Here we define some measures of distance. Note that we usually use $d(\cdot,\cdot)$ without mentioning the measure that we use, if its context is unambiguous. For distributions over $\{0,1\}^n$, $d(\cdot,\cdot)$ usually refers to the earth mover's distance.
	
	%\begin{definition}[A property]
	%    A \emph{property} $\mathcal{P}$ is a sequence $\mathcal{P}_1,\mathcal{P}_2,\ldots$ such that for every $n \ge 1$, $P_n$ is a compact subset of $\mathcal{D}(\{0,1\}^n)$.
	%\end{definition}
	
	%\begin{definition}[Distance of a distribution from a property]
	%    Let $\mathcal{P} = (\mathcal{P}_1,\mathcal{P}_2,\ldots)$ be a property and $P$ be a distribution over $\{0,1\}^n$ for some $n$. The \emph{distance of $P$ from $\mathcal{P}$} is defined as $d_\mathrm{EMD}(P,\mathcal{P}) = \min_{Q \in \mathcal{P}_n}\{d_\mathrm{EMD}(P,Q)\}$.
	%\end{definition}
	
	\subsection{Analysis of probabilistic algorithms} \label{sec:additional-prelims:subsec:algos}

	To be able to state and prove lower bounds, we use the notion of a ``distribution of runs''. Informally, it is the behavior of an algorithm $\mathcal{A}$ on an input that is drawn from a distribution. If we have a distribution over inputs in a property and another distribution over inputs that are $\eps$-far from the property, and their distributions over runs (with respect $\mathcal{A}$) are indistinguishable, then $\mathcal{A}$ cannot be an $\eps$-test for that property.
	
	\begin{definition}[Distribution of runs] \label{def-distribution-of-runs}
		Let $T$ be an $s$-sample, $q$-query decision tree and let $D$ be a distribution over inputs. The \emph{distribution of $T$-runs on $D$} is denoted by $\mathcal{R}(T,D)$, and is defined over $\{0,1\}^q$ as follows: first draw an input $P \sim D$. Then draw $s$ independent samples from $P$, and make the queries of $T$ (following the corresponding root to leaf path). The result is the vector of answers (of size $q$, padded with zeroes if necessary).
		
		Note that for non-adaptive algorithms, the distribution of runs can be seen as a distribution over functions from the fixed query set $Q$ to $\{0,1\}$, and can be obtained by drawing a distribution $P$, populating the matrix $M$ using $s$ independent samples from $P$, and then using the restriction $M|_Q$.
	\end{definition}
	Definition \ref{def-distribution-of-runs} can be naturally generalized to probabilistic algorithms since they can be seen as a distribution over deterministic ones. That is, $\mathcal{R}(\mathcal{A},D)$ distributes like ``draw $T \sim \mathcal{A}$ and then draw the run according to $\mathcal{R}(T,D)$''.
	
	\begin{definition}[Conditional distribution of runs]
		Let $T$ be an $s$-sample, $q$-query decision tree and let $D$ be a distribution over inputs. Also, let $B$ be some event. The \emph{distribution of $T$-runs on $D$ conditioned on $B$} is denoted by $\mathcal{R}\left(T,D \cond B\right)$, and is defined as the distribution $\mathcal{R}(T,D)$ restricted to $B$.
	\end{definition}
	
	\begin{lemma}[Lower bounds by Yao's principle] \label{lemma:yao-lbnd}
		Let $\mathcal{T}$ be a class of deterministic decision trees (which in turn define a class of probabilistic algorithms) and let $q > 0$. Let $D_1$ and $D_2$ be two distributions over inputs. If, for every decision tree $T \in \mathcal{T}$ of size less than $q$, $d(\mathcal{R}(T,D_1), \mathcal{R}(T,D_2)) < \frac{1}{3}$, then every probabilistic algorithm that distinguishes $D_1$ and $D_2$ (with error less than $\frac{1}{3}$) must make make at least $q$ queries (with positive probability).
	\end{lemma}
	
	The simulation method, described in the lemma below (whose proof is trivial), is a useful ``user interface'' for Yao's principle. %Its advantage over the standard-form of Yao's analysis in the Huge Object model is the ability to correlate samples of two input distributions.
	
	\begin{lemma}[The simulation method]\label{lemma:simulation-method-weak-form}
		Let $T$ be an $s$-sample, $q$-query decision tree, and let $D_1$, $D_2$ be distributions of inputs. Assume that there exist two events $B_1$ and $B_2$ for which $\mathcal{R}(T,D_1|\neg B_1)$ is identical to $\mathcal{R}(T,D_2|\neg B_2)$. In this setting, $d(\mathcal{R}(T,D_1), \mathcal{R}(T,D_2)) \le \Pr_{\mathcal{R}(T,D_1)}[B_1] + \Pr_{\mathcal{R}(T,D_2)}[B_2]$.
	\end{lemma}

	\subsection{Analysis of properties of distributions}

	\begin{definition}[Being closed under mapping, \cite{gr22}]
		A property $\mathcal{P}$ of distributions over $\{0,1\}^n$ is \emph{closed under mapping} if for every function $f : \{0,1\}^n \to \{0,1\}^n$ and for every distribution  $P \in \mathcal{P}$ we have $f(P) \in \mathcal{P}$.
	\end{definition}
	Note that the bounded support properties ($\mathcal{S}_m$ for fixed sizes and $\mathcal{S}_f$ for functions) are closed under mapping.
	
	\begin{definition}[Label-invariance]
		A property $\mathcal{P}$ of distributions over $\{0,1\}^n$ is \emph{label-invariant} if for every distribution $P \in \mathcal{P}$ and for every permutation $\sigma : \{0,1\}^n \to \{0,1\}^n$, $\sigma(P) \in \mathcal{P}$ as well.
	\end{definition}
	Note that every property that is closed under mapping is also label-invariant.
	
	The following proposition informally states that a specific strong constraint fully characterizes the bounded support property. We prove it in Appendix \ref{apx:proof-of-proposition}.
	
	\begin{restatable}{proposition}{propZlblZinvZandZoneZsideZisZsupp}
		\label{prop:lbl-inv-and-one-side-is-supp}
		Consider any label-invariant property of distributions $\mathcal{P}$ that has a one-sided $\eps$-test for every $\eps > 0$ (with any number of samples and queries). There exists a function $f : \mathbb{N} \to \mathbb{N}$ such that $\mathcal{P} = \mathcal{S}_f$.
	\end{restatable}

	\subsection{Analysis of the bounded support property}
	
	We start by stating simple observations that we use throughout our work.
	
	\begin{lemma}[\cite{adar23}] \label{lemma:realizing-the-distance-from-a-set}
		Let $P$ be a distribution over $\{0,1\}^n$ and $A$ be a subset of $\{0,1\}^n$. Let $f : \supp(P) \to A$ be the following function: for every $x \in \supp(P)$, $f(x) = \arg \min_{y \in A}\{d(x,y)\}$ (ties are broken arbitrarily). Then $d(P, f(P)) = d(P,\mathcal{S}_A)$.
	\end{lemma}
	
	\begin{observation}[\cite{adar23}] \label{obs:dxA-far}
		Let $A\subseteq\{0,1\}^n$ be a set of elements, and let $P$ be a distribution that is $\eps$-far from being supported by any subset of $A$. Then $d(P,\mathcal{S}_A) = \E\limits_{x \sim P}[d(x,A)] > \eps$.
	\end{observation}
	
	% contradiction graph
	
	%\begin{definition}[Contradiction graph] \label{def:contradiction-graph}
	%    Let $x_1,\ldots,x_s \in \{0,1\}^n$ be a sequence of strings. Let $Q \subseteq \{1\ldots,s\} \times \{1,\ldots,n\}$ be a set of queries. We define \emph{the contradiction graph} of $(x_1,\ldots,x_s;Q)$ as $G(V,E)$ with $V=\{1,\ldots,s\}$, and for every $1 \le i_1,i_2 \le s$:
	%    \[\{i_1,i_2\} \in E \iff \exists 1 \le j \le n : (x_{i_1})_j \ne (x_{i_2})_j \wedge \left((i_1,j),(i_2,j) \in Q\right)\]
	%    Note that the graph is undirected since the definition of the edges is commutative.
	%\end{definition}
	
	%[moved before the observation] \begin{lemma}[\cite{adar23}] \label{lemma:realizing-the-distance-from-a-set}
		%    Let $P$ be a distribution over $\{0,1\}^n$ and $A$ be a subset of $\{0,1\}^n$. Let $f : \supp(P) \to A$ be the following function: for every $x \in \supp(P)$, $f(x) = \arg \min_{y \in A}\{d(x,y)\}$ (ties are broken arbitrarily). Then $d(P, f(P)) = d(P,\mathcal{S}_A)$.
		%\end{lemma}
		
		%\begin{definition}[Witness against $m$-support] %\label{def:witness-against-Sm-not-colorable}
		%    Let $P$ be a distribution that is supported by a set of more than $m$ elements. We say that $(x_1,\ldots,x_s; Q)$ is a \emph{witness against $m$-support} (of $P$) if $x_1,\ldots,x_s$ are all drawn from $P$, and their contradiction graph is not $m$-colorable.
		%\end{definition}
		
		The following lemma shows the correctness of Definition \ref{def:witness-against-Sm-not-colorable} (witness against $m$-support).
		\begin{lemma} \label{lemma:witness-iff-not-colorable}
			Let $x_1,\ldots,x_s \in \supp(P)$ be a set of samples and let $Q \subseteq \{1,\ldots,s\} \times \{1,\ldots,n\}$ be a query set. Let $Q_1,\ldots,Q_s$ be the sample-specific query sets, that is, $Q = \bigcup_{i=1}^s (\{i\}\times Q_i)$, and let $G$ be the contradiction graph as per Definition \ref{def:contradiction-graph}. If $G$ is not colorable by $m$ colors, then $|\{x_1,\ldots,x_s\}| > m$. And if $G$ is colorable by $m$ colors, then there exists $\hat{P}$ with $|\supp(\hat{P})| \le m$ and a sequence $y_1,\ldots,y_s \in \supp(\hat{P})$ such that for every $1 \le i \le s$, $x_i|_{Q_i} = y_i|_{Q_i}$.
		\end{lemma}
		
		\begin{proof}
			Let $A \subseteq \{0,1\}^n$, and let $f : \{x_1,\ldots,x_m\} \to A$ be a mapping such that for every $1 \le i \le s$, $x_i|_{Q_i} = (f(i))|_{Q_i}$. If $G$ is not colorable by $m$ colors, then $f$ cannot be a valid coloring unless $|A| > m$. Specifically, $|\{x_1,\ldots,x_s\}| > m$, since with $A=\{x_1,\ldots,x_s\}$ we have the coloring $f(i) = x_i$.
			
			Now assume that $G$ is colorable by $m$ colors. Let $f : \{1,\ldots,s\} \to \{1,\ldots,m\}$ be a valid coloring. Let $\hat{A} = \{y_1,\ldots,y_m\}$ be the following set: for every $1 \le k \le m$, $y_k \in \{0,1\}^n$ is a string for which $y_k|_{Q_i} = x_i|_{Q_i}$ for every $i$ for which $f(i) = k$. A concrete example for $y_k$ would be: the $j$-th bit is $1$ if and only if there exists $1 \le i \le s$ such that $f(i) = k$, $j \in Q_i$ and $x_i|_j = 1$. Let $\hat{P}$ be the uniform distribution over this $\hat{A}$. There exist $y_1,\ldots,y_s \in \supp(\hat{P})$ such that for every $1 \le i \le s$, $x_i|_{Q_i} = y_i|_{Q_i}$. Finally, note that $|\supp(\hat{P})| \le m$.
		\end{proof}
		
		%\begin{definition}[Explicit witness against $m$-support]
		%    Let $P$ be a distribution that is supported by a set of more than $m$ elements. We say that $(x_1,\ldots,x_s,Q)$ is an \emph{explicit witness against $m$-support} (of $P$) if $x_1,\ldots,x_s$ are all drawn from $P$, and their contradiction graph contains a clique with $m+1$ vertices as a subgraph.
		%\end{definition}
		
		The following lemma is a counterpart of Lemma \ref{lemma:yao-lbnd} for the special case of one-sided error. It follows from the same observation by Yao that a probabilistic algorithm can be viewed as a distribution over deterministic algorithms, along with the observation that without loss of generality a one-sided error algorithm rejects if and only if it finds a witness against the property.
		
		\begin{lemma}[Lower bounds by Yao's principle for one-sided algorithms] \label{lemma:yao-lbnd-one-sided}
			Let $\mathcal{P}$ be a property, $\mathcal{T}$ be a class of deterministic decision trees and let $q > 0$. Let $D$ be a distribution over inputs that always draws an input distribution that is $\eps$-far from $\mathcal{P}$.
			
			Consider a decision tree $T \in \mathcal{T}$, and let $\mathcal{B}_T$ be the set of witnesses against $\mathcal{P}$ (that is, the set of unreachable leaves, or runs, for any input $P \in \mathcal{P}$). If
			$\Pr_{\mathcal{R}(T,D)}[\mathcal{B}_T] < \frac{1}{3}$, for every $T \in \mathcal{T}$ of size less than $q$, then every one-sided probabilistic algorithm for $\mathcal{P}$ conforming to $\mathcal{T}$ must make at least $q$ queries (with positive probability).
		\end{lemma}
		
		\section{Superlinear lower-bound for non-adaptive $2$-support test}
		\label{sec:2-supp-lbnd}
		We show an $\Omega(\eps^{-1} \log \eps^{-1})$ lower bound for non-adaptive $2$-support tests, even with two-sided error. Generalizing the construction, we show a bound of $\Omega(\eps^{-1} \log \eps^{-1})$ for $m$-support tests with any $m \ge 2$ (note that a single $f(\eps) = \Omega(\eps^{-1} \log \eps^{-1})$ holds simultaneously for all $m \ge 2$, rather than having an implicit coefficient that depends on $m$), and a one-sided bound of $\Omega(\eps^{-1} \log \eps^{-1} \cdot m)$.
		
		\begin{theorem} \label{th:m-supp-na-lbnd-two-sided}
			Every non-adaptive $\eps$-test for $\mathcal{S}_m$ must make $\Omega(\eps^{-1} \log \eps^{-1})$ queries, even if it has two-sided error.
		\end{theorem}
		
		We prove this theorem in this section. To do so, we first define distributions over inputs (which are in themselves distributions over $\{0,1\}^n$) and analyze them.
		
		\begin{definition}[$D^{t}_\mathrm{no}$ for $\eps$]
			\label{def:dno-t}
			Draw $\alpha$ such that $\log_2 \alpha^{-1}$ is uniform over $\{2,\ldots,\floor{\log_2 \eps^{-1}} - 2\}$.
			Draw a set $D \subseteq \{1,\ldots,n\}$ such that for every $1 \le j \le n$, $\Pr[j \in D] = 4\alpha$, independently. Then, for every $1 \le k \le t$, draw a set $A_k \subseteq D$ such that for every $j \in D$, $\Pr[j \in A_k | j \in D] = \frac{1}{2}$, independently. The resulting input is defined as the following distribution over $\{0,1\}^n$:
			\begin{align*}
				P : \begin{cases}
					0 & \mathrm{with\ probability\ } 1 - 2\alpha^{-1} \eps \\
					1_{A_1} & \mathrm{with\ probability\ } 2\alpha^{-1} \eps / t \\
					\vdots \\
					1_{A_t} & \mathrm{with\ probability\ } 2\alpha^{-1} \eps / t
				\end{cases}
			\end{align*}
		\end{definition}
		
		\begin{definition}[$D_\mathrm{yes}$ for $\eps$]
			\label{def:dyes-s2}
			We simply define $D_\mathrm{yes}$ as $D^1_\mathrm{no}$. An equivalent definition of $D_\mathrm{yes}$ is the following: draw $\alpha$ such that $\log_2 \alpha^{-1}$ is uniform over $\{2,\ldots,\floor{\log_2 \eps^{-1}} - 2\}$, and then draw a set $A \subseteq \{1,\ldots,n\}$ such that for every $1 \le j \le n$, $\Pr[j \in A] = 2\alpha$, independently. The resulting input is defined as the following distribution over $\{0,1\}^n$:
			\begin{align*}
				P : \begin{cases}
					0 & \mathrm{with\ probability\ } 1 - 2\alpha^{-1} \eps \\
					1_A & \mathrm{with\ probability\ } 2\alpha^{-1} \eps
				\end{cases}
			\end{align*}
		\end{definition}
		
		First we show that $D^t_\mathrm{no}$ and $D_\mathrm{yes}$ can be used to demonstrate lower bounds for $\eps$-testing $\mathcal{S}_m$. Trivially, $D_\mathrm{yes}$ draws a distribution in $\mathcal{S}_2$ (and hence $\mathcal{S}_m$) with probability $1$.
		\begin{observation}\label{obs:sum-of-halves}
			If $a_1,\ldots,a_k$ are non-negative integers, then $\sum_{i=1}^k \floor{a_i / 2} \ge \frac{\sum_{i=1}^k a_i -k+1}{2}$.
		\end{observation}
		
		\begin{lemma} \label{lemma:dno-eps-far-sm}
			Let $t \ge 2$ and $P \sim D^{t}_\mathrm{no}$. For sufficiently large $n$ (as a function of $t$ and $\eps$) and for every $1 \le m \le t$, with probability at least $1 - \frac{3}{1000}$, the distance of $P$ from $\mathcal{S}_m$ is more than $(2 - \frac{2m - 2}{t})\eps$.
			
			Concretely, $D^2_\mathrm{no}$ is $\eps$-far from $\mathcal{S}_2$ with probability at least $1 - \frac{3}{1000}$, and for every $t \ge 3$, $P$ is $\eps$-far from $\mathcal{S}_{\ceil{t/2}}$ with this probability.
		\end{lemma}
		
		\begin{proof}
			Let $n > 480 000 \eps^{-1} t^2 (\ln t + 10)$. By the multiplicative Chernoff's bound (Lemma \ref{lemma:chernoff-mult}), every individual event of one of the following forms happens with probability at least $1 - \frac{3}{1000 t^2}$:
			\begin{itemize}
				\item $\left(4 - \frac{1}{100t}\right)\alpha n < |D| < \left(4 + \frac{1}{100t}\right)\alpha n$.
				\item $\left(2 - \frac{1}{100t}\right)\alpha n < |A_k| < \left(2 + \frac{1}{100t}\right)\alpha n$, for every $1 \le k \le t$.
				\item $\left(2 - \frac{1}{100t}\right)\alpha n < |A_{k_1} \Delta A_{k_2}| < \left(2 + \frac{1}{100t}\right)\alpha n$, for every $1 \le k_1 < k_2 \le t$.
			\end{itemize}
			By the union bound over $\binom{t}{2} + t + 1 \le t^2$ events, with probability $1 - \frac{3}{1000}$, all of the above events happen simultaneously.
			
			Assume that this is the case, and consider a set $A$ of size $m$. Let $f : \supp(P) \to A$ be a mapping from every element in the support of $P$ to a closest element in $A$ (ties are broken arbitrarily). By Lemma \ref{lemma:realizing-the-distance-from-a-set}, $f$ realizes the distance from $P$ to a closest distribution supported by $A$, that is, $d(P,f(P)) = d(P, \mathcal{S}_A)$. Consider some $a \in A$. For every $u,v \in f^{-1}(a)$, the contribution to the distance is at least $d(u,a) + d(v,a) \ge d(u,v) > \left(2 - \frac{1}{100t}\right)\alpha$. The probability weight of $u,v$ is at least $2\alpha^{-1}\eps / t$ individually, hence their contribution to the distance is more than $\left(4 - \frac{1}{50t}\right)\eps/t$.
			
			Considering $\floor{|f^{-1}(a)| / 2}$ disjoint pairs of elements in $f^{-1}(a)$, the contribution of $f^{-1}(a)$ elements to the distance $d(P,\mathcal{S}_A)$ is at least $(4 - \frac{1}{50t})\frac{\eps}{t} \floor{|f^{-1}(a)| / 2}$. Summing over all $a \in A$ we have:
			\begin{eqnarray*}
				d(P,\mathcal{S}_A)
				> \frac{\left(4 - \frac{1}{50t}\right)\eps}{t} \sum_{a \in A} \floor{|f^{-1}(a)| / 2}
				&\underset{(*)}\ge& \frac{(4 - \frac{1}{50t})\eps}{t} \cdot \frac{|\supp(P)| - |A| + 1}{2} \\
				&=& \left(2 - \frac{1}{100t}\right) \frac{\eps}{t}\left(t + 2 - m\right)
				\ge \left(2 - \frac{2m - 2}{t}\right)\eps
			\end{eqnarray*}
			The starred transition is implied by Observation \ref{obs:sum-of-halves}. This holds for every $A$ of size $m$, hence $d(P,\mathcal{S}_m) > (2 - \frac{2m - 2}{t})\eps$.
		\end{proof}
		
		We show that a non-adaptive algorithm that only uses $q < \frac{1}{170} \eps^{-1} \log \eps^{-1}$ queries cannot distinguish $D^t_\mathrm{no}$ from $D_\mathrm{yes}$ with error smaller than $\frac{2}{7} < \frac{1}{3} - \frac{3}{1000}$, for every $t \ge 2$ and sufficiently small $\eps$ (whose bound is independent of $t$).
		
		\begin{proof}[Proof (of Theorem \ref{th:m-supp-na-lbnd-two-sided})]
			Let $Q$ be a query set of size $q < \frac{1}{300} \eps^{-1} \log \eps^{-1}$. Our $D_\mathrm{no}$ for this proof is $D^t_\mathrm{no}$ with $t=2m$ (and $D_{\mathrm{yes}}=D^1_{\mathrm{no}}$). By Lemma \ref{lemma:dno-eps-far-sm}, a distribution that is drawn from $D^{2m}_\mathrm{no}$ is $(2 - \frac{2m-2}{t})\eps$-far from $\mathcal{S}_m$ with probability $\frac{3}{1000}$, which is $\eps$-far for every $m \ge 2$.
			
			Let $R^Q$ be the following distribution over responses to the query set:
			\begin{itemize}
				\item Choose $\alpha$ such that $\log \alpha^{-1}$ is uniform in $\{2,\ldots,\floor{\log \eps^{-1}} - 2\}$.
				\item Choose a set $D \subseteq \{1,\ldots,n\}$ such that for every $1 \le j \le n$, $\Pr[j \in A] = 4\alpha$, independently.
				\item Choose a set $S \subseteq \{1,\ldots,s\}$ such that for every $1 \le i \le s$, $\Pr[i \in S] = 2\alpha^{-1}\eps$, independently.
				\item The result is $f : Q \to \{0,1\}$, where $f(i,j) = 0$ if $i \notin S$ or $j \notin D$, and for every $(i,j) \in S \times D$, $\Pr[f(i,j) = 1] = \frac{1}{2}$, independently.
			\end{itemize}
			
			We show that both $\mathcal{R}(Q,D_\mathrm{yes})$ and $\mathcal{R}(Q,D_\mathrm{no})$ are $\frac{1}{7}$-close to a distribution related to $R^Q$, hence they must be $\frac{2}{7}$-close to each other. By Yao's method (Lemma \ref{lemma:yao-lbnd}), this means that there is no non-adaptive $\eps$-test for $\mathcal{S}_m$ making only $q$ queries, as required.
			
			Let $\mathcal{B}$ be the following bad event: there exists an index $j \in D$ at which two (or more) non-zero samples are queried. If $\mathcal{B}$ does not happen then the algorithm does not even have an opportunity to compare those non-zero samples (which in mathematical terms means having the same distribution over the query answers when conditioned on this event). In $R^Q$, $\mathcal{B}$ is defined correspondingly as the event that there exists $j \in D$, $i_1, i_2 \in S$ ($i_1 \ne i_2$) for which $(i_1,j), (i_2,j) \in Q$.
			
			The rest of the proof has two parts: the first is proving that the probability of $\mathcal{B}$ is less than $\frac{1}{3} - \frac{3}{1000}$; and the second is proving that $D_\mathrm{no}$ and $D_\mathrm{yes}$ are both identical to $R^Q$ when those are conditioned on the negation of $\mathcal{B}$.
			
			We decompose $Q$ by its indices. That is, $Q = \bigcup_{j=1}^n \{j\} \times S_j$, where $S_j$ are the samples that the algorithm queries at the $j$-th index. For every $2 \le \ell \le q$, let $w_{\ell} = \left|\{ 1 \le j \le n : |S_j| = \ell \}\right|$ be the number of indices that have exactly $\ell$ samples.
			
			Consider some index $j$ with $\ell$ samples. Given $\alpha$, the probability to draw more than one non-zero sample among these $\ell$ samples is bounded by $\min\{1,(2 \alpha^{-1} \eps \ell)^2\} $. Note that the non-constant expression is effective only when $\log \alpha^{-1} \le \log \eps^{-1} - \log \ell - 1$.
			
			For every $1 \le j \le n$, let $X_j$ be an indicator for being queried in two (or more) non-zero samples, as well as belonging to $D$. If $X_j = 0$, then the algorithm cannot distinguish between any two non-zero samples using $j$-queries. Consider some index $j$ with $|S_j| = \ell$ samples.
			
			\begin{eqnarray*}
				\E[X_j] &=& \Pr[X_j = 1]
				= \sum_{a=2}^{\floor{\log \eps^{-1}} - 2} \Pr\left[\alpha = 2^{-a}\right] \Pr\left[j \in D \cond \alpha = 2^{-a}\right] \Pr\left[X_j = 1 \cond \alpha = 2^{-a}, j \in D\right] \\
				&\le& \frac{1}{\floor{\log \eps^{-1}} - 3} \sum_{a=2}^{\floor{\log \eps^{-1}} - 2} 4\cdot 2^{-a} \min\{1, (2 \cdot 2^a \eps \ell)^2\} \\
				&=& \frac{1}{\floor{\log \eps^{-1}} - 3}\left(\sum_{a=2}^{\floor{\log \eps^{-1} - \log \ell} - 1} 4\cdot 2^{-a} (2 \cdot 2^a \eps \ell)^2 + \sum_{a=\floor{\log \eps^{-1} - \log \ell}}^{\floor{\log \eps^{-1}} - 2} 4\cdot 2^{-a} \right) \\
				&=& \frac{1}{\floor{\log \eps^{-1}} - 3}\left(
				\left(16\eps^2 \ell^2 \sum_{a=2}^{\floor{\log \eps^{-1} - \log \ell} - 1} 2^a \right) +
				\left(4 \sum_{a=\floor{\log \eps^{-1} - \log \ell}}^{\floor{\log \eps^{-1}} - 2} 2^{-a}\right)
				\right) \\
				&\le& \frac{1}{\floor{\log \eps^{-1}} - 3}\left(
				\left(16\eps^2 \ell^2 \cdot \eps^{-1} / \ell \right) +
				\left(4 \cdot 2\ell \eps \right)
				\right)
				\le \frac{24\eps \ell}{\floor{\log \eps^{-1}} - 3}
			\end{eqnarray*}
			
			Let $X$ be the number of indices that are queried in two (or more) non-zero samples, and also belong to $D$. Considering all indices, by linearity of expectation:
			\begin{eqnarray*}
				\E[X]
				= \sum_\ell \sum_{j : |S_j| = \ell} \E[X_j]
				&\le& \sum_\ell \frac{24\eps \ell}{\floor{\log \eps^{-1}} - 3} w_\ell \\
				&<& \frac{24\eps}{\floor{\log \eps^{-1}} - 3} \cdot \frac{1}{300} \eps^{-1} \log \eps^{-1}
				= \frac{2}{25} \cdot \frac{\log \eps^{-1}}{\floor{\log \eps^{-1}} - 3}
				\underset{\eps < 2^{-12}}< \frac{1}{7}
			\end{eqnarray*}
			To summarize the first part of the proof, $\Pr[\mathcal{B}] = \Pr[X \ne 0] \le \E[X] < \frac{1}{7}$ for $\eps < 2^{-12}$. Note that at this point we showed a lower bound of $\Omega(\eps^{-1} \log \eps^{-1})$ queries for a one-sided $\eps$-test of $\mathcal{S}_m$.
			
			In the following we show that if $\mathcal{B}$ did not happen then $D_\mathrm{no}$ and $D_\mathrm{yes}$ are identical to $R^Q$ conditioned on $\neg\mathcal{B}$. Note that $D_\mathrm{yes}$ always draws a distribution in $\mathcal{S}_m$, since $\mathcal{S}_2 \subseteq \mathcal{S}_m$ for every $m \ge 2$.
			
			Consider the answer function of the run. For every query $(i,j) \in Q$, if the $i$-th sample is a zero sample (that is, $i \notin S$), then $f(i,j) = 0$. Additionally, if $j \notin D$, then $f(i,j) = 0$ as well. For $(i,j) \in S \times D$, $\Pr[f(i,j) = 1 | i \in S, j \in D] = \Pr[(1_{A_{k_j}})_j = 1 | j \in D] = \frac{1}{2}$. Also, if $\mathcal{B}$ does not happen, then for every $(i,j) \ne (i',j')$ that are both in $S \times D$ we must have $j \ne j'$. Thus all the ``$f(i,j)$'' events for $Q \cap (S \times D)$ are mutually independent, making the distribution of the answers to the queries, when conditioned on $\neg\mathcal{B}$, identical to $R^Q$ conditioned on $\neg \mathcal{B}$.
			
			The above argument holds for both $D_\mathrm{yes}$ and $D_\mathrm{no}$, hence $\mathcal{R}(Q,D_{\mathrm{yes}})$ and $\mathcal{R}(Q,D_{\mathrm{no}})$ are $\frac{2}{7}$-close to each other (since $2\Pr[\mathcal{B}] \le \frac{2}{7}$). Since $D_\mathrm{no}$ draws a distribution $\eps$-far from $\mathcal{S}_m$ with probability $\frac{3}{1000}$, and $D_\mathrm{yes}$ always draws a distribution that belongs to $\mathcal{S}_m$, every non-adaptive $\eps$-test for $\mathcal{S}_m$ must make at least $\Omega(\eps^{-1} \log \eps^{-1})$ queries.
		\end{proof}
		
		\subsection{Composite lower bound for non-adaptive one-sided $\eps$-tests of $\mathcal{S}_m$}
		\begin{theorem} \label{th:m-supp-na-one-sided-lbnd-composite}
			Every one-sided non-adaptive $\eps$-test of $\mathcal{S}_m$ must make $\Omega(\eps^{-1} m \log \eps^{-1})$ queries.
		\end{theorem}
		
		\begin{proof}
			Without loss of generality, assume that $m$ is divisible by $3$. Our $D_\mathrm{no}$ for this proof is $D^t_\mathrm{no}$ with $t = \frac{4}{3}m$ and $\hat{\eps} = 2\eps$. By Lemma \ref{lemma:dno-eps-far-sm}, with probability $1 - o(1)$, $D_\mathrm{no}$ draws a distribution that is $(\frac{1}{2} + \frac{3}{4m})\hat{\eps}$-far from $\mathcal{S}_m$, which is $\eps$-far. For a bound against a one-sided test we will bound the probability of finding a witness against $\mathcal{S}_m$ under this distribution (and then use Lemma \ref{lemma:yao-lbnd-one-sided}), so in particular there is no $D_{\mathrm{yes}}$.
			
			Let $Q$ be a query set of size $q < \frac{m}{1000}\eps^{-1} \log \eps^{-1} \le \frac{m}{500}\hat{\eps}^{-1} \log \hat{\eps}^{-1}$. Fix two indices $1 \le k_1 < k_2 \le \frac{4}{3}m$. For every $1 \le j \le n$, let $X_{j,k_1,k_2}$ be an indicator for having queries in two (or more) samples in $\{1_{A_{k_1}}, 1_{A_{k_2}}\}$, as well as belonging to $D$. If $X_{j,k_1,k_2} = 0$, then the algorithm cannot distinguish between $1_{A_{k_1}}$ and $1_{A_{k_2}}$ using $j$-queries. Also, let $X_{k_1,k_2} = \sum_{j=1}^n X_{j,k_1,k_2}$ be the number of indices that have an opportunity to distinguish between $1_{A_{k_1}}$ and $1_{A_{k_2}}$ (note that it is similar to $X$ in Theorem \ref{th:m-supp-na-lbnd-two-sided}). Finally let $\mathcal{B}_{k_1,k_2}$ be the event for $X_{k_1,k_2} > 0$, and $Y_{k_1,k_2} \in \{0,1\}$ be the indicator for $\mathcal{B}_{k_1,k_2}$. The number of edges in the contradiction graph (Definition \ref{def:contradiction-graph}), excluding the vertex representing the zero vector, is bounded by $\sum_{1 \le k_1 < k_2 \le \frac{4}{3}m} Y_{k_1,k_2}$ (if $Y_{k_1,k_2}=1$ then the algorithm has an \emph{opportunity} to distinguish them, but it can still fail to do so).
			
			Consider some $1 \le k_1 < k_2 \le \frac{4}{3}m$ and some index $j$ of with $\ell$ samples. The probability to draw a single sample in $\{1_{A_{k_1}}, 1_{A_{k_2}}\}$ is $2\cdot (2\alpha^{-1}\hat{\eps} / (\frac{4}{3}m)) = 3\alpha^{-1}\hat{\eps} / m$. The probability to draw more than one sample in $\{1_{A_{k_1}}, 1_{A_{k_2}}\}$ among these $\ell$ samples is bounded by $\min\{1, (3 \alpha^{-1} \eps \ell / m)^2\}$. As in the previous subsection, for every $2 \le \ell \le q$, let $w_{\ell} = \left|\{ 1 \le j \le n : |S_j| = \ell \}\right|$ be the number of indices that have exactly $\ell$ samples.
			
			For convenience, let $h = \log \hat{\eps}^{-1} + \log m - \log \ell$.
			
			\begin{eqnarray*}
				\E[X_{j,k_1,k_2}] &=& \Pr[X_{j,k_1,k_2} = 1] \\
				&=& \sum_{a=2}^{\floor{\log \hat{\eps}^{-1}} - 2} \Pr\left[\alpha = 2^{-a}\right] \Pr\left[j \in D \cond \alpha = 2^{-a}\right] \Pr\left[X_{j,k_1,k_2} = 1 \cond \alpha = 2^{-a}, j \in D\right] \\
				&\le& \frac{1}{\floor{\log \hat{\eps}^{-1}} - 3} \sum_{a=2}^{\floor{\log \hat{\eps}^{-1}} - 2} 4\cdot 2^{-a} \min\{1, (3 \cdot 2^a \hat{\eps} \ell / m)^2\} \\
				&=& \frac{1}{\floor{\log \hat{\eps}^{-1}} - 3}\left(\sum_{a=2}^{\floor{h - \log 3}} 4\cdot 2^{-a} (3 \cdot 2^a \hat{\eps} \ell / m)^2 + \sum_{a=\floor{h - \log 3} + 1}^{\floor{\log \hat{\eps}^{-1}} - 2} 4\cdot 2^{-a} \right) \\
				&=& \frac{1}{\floor{\log \hat{\eps}^{-1}} - 3}\left(
				\left(\frac{36}{m^2}\hat{\eps}^2 \ell^2  \sum_{a=2}^{\floor{h - \log 3}} 2^a \right) +
				\left(4 \sum_{a=\floor{h - \log 3} + 1}^{\floor{\log \hat{\eps}^{-1}} - 2} 2^{-a}\right)
				\right) \\
				&\le& \frac{1}{\floor{\log \hat{\eps}^{-1}} - 3}\left(
				\left(\frac{36}{m^2} \hat{\eps}^2 \ell^2 \cdot \frac{2}{3}\hat{\eps}^{-1} m / \ell \right) +
				\left(\frac{8}{3} \cdot \ell \hat{\eps} / m \right)
				\right) \\
				&\le& \frac{27 \hat{\eps} \ell}{m(\floor{\log \hat{\eps}^{-1}} - 3)}
			\end{eqnarray*}
			
			Considering all indices,
			\begin{eqnarray*}
				\E[X_{k_1,k_2}]
				&=& \sum_\ell \sum_{j : |S_j| = \ell} \E[X_{j,k_1,k_2}] \\
				&\le& \sum_\ell \frac{27\hat{\eps} \ell}{m(\floor{\log \hat{\eps}^{-1}} - 3)} w_\ell \\
				&<& \frac{27\hat{\eps}}{m(\floor{\log \hat{\eps}^{-1}} - 3)} \cdot \frac{m}{500} \hat{\eps}^{-1} \log \hat{\eps}^{-1}
				= \frac{27}{500} \cdot \frac{\log \hat{\eps}^{-1}}{\floor{\log \hat{\eps}^{-1}} - 3} < \frac{1}{16}
			\end{eqnarray*}
			(The last transition is correct for every sufficiently small $\hat\eps$). That is, $\Pr[Y_{k_1,k_2} = 1] = \Pr[\mathcal{B}_{k_1,k_2}] < \frac{1}{16}$ for every $1 \le k_1 < k_2 \le \frac{4}{3}m$.
			
			By linearity of expectation, the expected number of edges in the contradiction graph, excluding the vertex representing the zero vector, is bounded by $\frac{1}{16}\binom{4m/3}{2} < \frac{1}{4} \binom{2m/3}{2}$. By Markov's inequality, with probability at least $\frac{3}{4}$, this contradiction graph has less than $\binom{2m/3}{2}$ edges, hence it is colorable using $\frac{2}{3}m$ colors. Considering the zero vector as well, the contradiction graph can be colored using $\frac{2}{3}m + 1 \le m$ colors. In this case, there is no witness against $P \in \mathcal{S}_m$. The lower bound is then implied by Lemma \ref{lemma:yao-lbnd-one-sided}.
		\end{proof}
		
		\section{Quasilinear non-adaptive one-sided $m$-support test} 
		\label{sec:ubnd-na-one-sided}
		
		We show a one-sided non-adaptive $\eps$-test algorithm for $\mathcal{S}_m$ using $O(\eps^{-1} m \log \eps^{-1} \log m)$ queries for every $m \ge 2$. Note that for every fixed constant $\eps$ this bound is tight, since in Section \ref{sec:lbnd-adaptive-one-sided} we show an $\Omega(m \log m)$ lower bound. The bound is tight for every fixed constant $m$ as well, since we have a corresponding non-adaptive lower bound of $\Omega(\eps^{-1} \log \eps^{-1})$.
		
		Let $\eps > 0$ and $m \ge 2$. The algorithm looks for a set $A$ for of size at least $m+1$ whose elements are fully distinguishable using queries. The algorithm is defined for $\eps$ that is a power of $2$ (for other choices of $\eps$, we can use $\hat\eps = 2^{-\ceil{\log_2 \eps^{-1}}}$ instead).
		
		At first, the algorithm chooses $I_0 \subseteq I_1 \subseteq \cdots \subseteq I_{\log\eps^{-1}}\subseteq \{1,\ldots,n\}$, where $I_a$ consists of $\ceil{2^{a+2} \log (m+1)}$ indices drawn uniformly and independently.
		
		The algorithm takes $1 + 32\eps^{-1}m$ samples. Except for the first sample, they are partitioned into $2m$ ``blocks'' of at most $16\eps^{-1}$ samples each. For every $1 \le k \le 2m$ and $0 \le a \le \log\eps^{-1}$, the algorithm takes a sequence $S_{a,k}$ of $2^{3-a} \eps^{-1}$ new samples, and queries every sample in it at the indices of $I_a$.
		
		The algorithm rejects if there exists a \emph{distinguishable composition} of size $m+1$ (which in particular is also a witness against $\mathcal{S}_m$). We formally define this term (and others) below.
		
		\begin{definition}[$K$-composition] \label{def:composition}
			For $K \subseteq \{1,\ldots,2m\}$, a sequence $A = (u_1,a_2,u_2,\ldots,a_\ell,u_\ell)$ is called a \emph{$K$-composition of length $\ell$} if $u_1=u$ (the first sample) and for every $2 \le i \le \ell$, $0 \le a_i \le \log\eps^{-1}$ and $u_i \in S_{a_i,k_i}$ (for some $k_i \in K$). A $\{1,\ldots,2m\}$-composition is called a \emph{composition}.
		\end{definition}
		
		\begin{definition}[Soundness of a $K$-composition] \label{def:composition-soundness}
			For some $K$, let $A = (u_1,a_2,u_2,\ldots,a_\ell,u_\ell)$ be a $K$-composition. We say that it is \emph{sound}, if for every $1 \le i < j \le \ell$, $d(u_i, u_j) > 2^{-a_j-1}$.
		\end{definition}
		In other words, a composition is \emph{sound} if for every $i$, the choice of $a_i$ results in lower bound for the distance of $u_i$ from all $\{u_1,\ldots,u_{i-1}\}$. Note that the algorithm cannot be certain about the soundness of a $K$-composition unless it makes $\Omega(\eps n)$ queries for every individual element in the composition (which it does not).
		
		\begin{definition}[Monotonicity of a $K$-composition] \label{def:composition-monotonity}
			For some $K$, let $A = (u_1,a_2,u_2,\ldots,a_\ell,u_\ell)$ be a $K$-composition. We say that it is \emph{monotone} if $a_2 \ge \ldots \ge a_\ell$.
		\end{definition}
		Observe that the algorithm can easily verify the monotonicity of a $K$-composition.
		
		\begin{definition}[Distinguishability of a $K$-composition] \label{def:composition-distinguishability}
			For some $K$, let $A = (u_1,a_2,u_2,\ldots,a_\ell,u_\ell)$ be a $K$-composition. We say that it is \emph{distinguishable} if, for every $1 \le i_1 < i_2 \le \ell$, there exists a query $j \in I_{a_{i_1}} \cap I_{a_{i_2}}$ for which $(u_{i_1})_j \ne (u_{i_2})_j$. For this definition, we set $a_1 = \log \eps^{-1}$ (since $u_1 = u$, and the algorithm queries $u$ at $I_{\log \eps^{-1}}$).
		\end{definition}
		In other words, a sequence is distinguishable if for every two samples in the composition, there exists a common query that distinguishes them. Observe that the algorithm can always verify the distinguishability of a $K$-composition.
		
		\begin{definition}[Valid composition] \label{def:composition-valid}
			For some $K$, a $K$-composition is \emph{valid} if it is both sound and monotone.
		\end{definition}
		
		\begin{definition}[Rank of a monotone composition]
			Let $A = (u_1,a_2,u_2,\ldots,a_\ell,u_\ell)$ be a monotone $K$-composition (for some $K$). Its \emph{rank} is defined as $\vec{r}(A) = (a_2,\ldots,a_\ell)$.
		\end{definition}
		
		\begin{definition}[Order of ranks]
			Ranks are ordered lexicographically as strings in $\{0,\ldots,\log\eps^{-1} \}^*$. In particular, if $\vec{r}(A_1)$ is a proper prefix of $\vec{r}(A_2)$, then $\vec{r}(A_1) < \vec{r}(A_2)$. That is, the ``end-of-string'' virtual character is considered smaller than every actual character.
		\end{definition}
		
		\begin{algorithm}
			\caption{Non-adaptive construction of a valid composition}
			\label{alg:m-supp-na}
			\begin{algorithmic}
				\State \textbf{choose} indices $i_1,\ldots,i_{\ceil{4\eps^{-1}\log (m+1)}}$ uniformly and independently, with repetitions.
				\For{$0 \le a \le \log\eps^{-1}$}
				\Let $I_a = \{i_1, \ldots, i_{\ceil{2^{a+2} \log (m+1)}}\}$.
				\EndFor
				\State \textbf{take} a sample $u$.
				\State \textbf{query} $u$ at $I_{\log\eps^{-1}}$.
				\For{$k$ \textbf{from} $1$ \textbf{to} $2m$}
				\For{$a$ \textbf{from} $0$ \textbf{to} $\log\eps^{-1}$}
				\State \textbf{take} $2^{3-a}\eps^{-1}$ new samples, denoting the sequence by $S_{a,k}$.
				\State \textbf{query} all samples in $S_{a,k}$ at $I_a$.
				\EndFor
				\EndFor
				\If{there exists a distinguishable composition of size $m+1$}
				\State \Return \reject
				\Else
				\State \Return \accept
				\EndIf
			\end{algorithmic}
		\end{algorithm}
		
		\begin{observation}
			For every composition $A$ of length $\ell$ there exists some $K \subseteq \{1,\ldots,2m\}$ of size $\ell-1$ for which $A$ is a $K$-composition.
		\end{observation}
		
		\begin{definition}[Bad events]
			For every $1 \le \ell \le m$ and for every $K \subseteq \{1,\ldots,2m\}$ of size $\ell$, let $B_{K,\ell}$ be the following event: considering the maximum rank $\vec{r}$ of a valid $K$-composition $A = (u_1, a_2, u_2, \ldots, a_\ell, u_\ell)$ of length $\ell$ (for some $a_2, \ldots, a_\ell$), there is no valid composition $\tilde{A}$ (not necessarily a $K$-composition) with $\vec{r}(\tilde{A}) > \vec{r}(A)$ of size $2 \le \ell' \le \ell + 1$.
		\end{definition}
		
		Informally, $B_{K,\ell}$ is the event that there exist some valid $K$-composition whose length is too short and yet has the maximal rank among all valid compositions.
		
		\begin{lemma}\label{lemma:maximum-must-be-long}
			Assume that for every $1 \le \ell \le m$ and for every $K \subseteq \{1,\ldots,2m\}$ of size $\ell$, $B_{K,\ell}$ does not happen. Then there exists a valid composition of size $m+1$ or more.
		\end{lemma}
		\begin{proof}
			Let $A$ be a valid composition with the maximal rank among all valid compositions ($A$ is not necessarily unique). Let $K \subseteq \{1, \ldots, 2m\}$ be the set of blocks that contain the elements of $A$ ($|K| \le |A|$). If $|A| \le m$, then by maximality, $B_{K,|A|}$ must have happened. Since we assumed that it has not, the length of $|A|$ must be at least $m + 1$.
		\end{proof}
		
		\begin{lemma}\label{lemma:always-better-composition}
			If the input distribution is $\eps$-far from $\mathcal{S}_m$, then for every $1 \le \ell \le m$ and $K \subseteq \{1,\ldots,2m\}$ of size $\ell$, $\Pr[B_{K,\ell}] \le e^{-4m}$.
		\end{lemma}
		\begin{proof}
			There are at least $2m -|K| \ge m$ blocks that are entirely free from the conditions on $A$. For every $k \in \{1, \ldots, 2m\} \setminus K$, and for every $0 \le a \le \log \eps^{-1}$, the expected number of samples in $S_{a,k}$ that are $2^{-a-1}$-far from all elements $A$ is at least 
			\[\Pr\left[2^{-a-1} < d(x,A) \le 2^{-a}\right] \cdot 2^{3-a}\eps^{-1} \ge 2^a c^A_a \cdot 8 \cdot 2^{-a}\eps^{-1} = 8\eps^{-1}c^A_a\]
			Where $c^A_a$ is a short notation for $\Ct[d(x,A) | 2^{-a-1} < d(x,A) \le 2^{-a}]$.
			
			Considering all possible values for $a$ in the same block, the expected number of these ``matches'' is at least
			\[\sum_{a=0}^{\log\eps^{-1}} 8\eps^{-1} c^{A}_a
			= 8 \eps^{-1} \Ct\left[d(x,A) \cond d(x,A) > \frac{1}{2}\eps \right]
			\ge 8\eps^{-1}\left(\E[d(x,A)] - \frac{1}{2}\eps\right)
			> 8 \eps^{-1} \cdot \frac{1}{2}\eps = 4\]
			
			Considering all $k\in\{1,\ldots,2m\}\setminus K$ as well (at least $m$ of them), the expected number of these matches is at least $4m$. This is a sum of independent binomial variables, hence by Lemma \ref{lemma:pr-chernoff-zero} the probability that there are no matches at all is bounded by $e^{-4m}$. That is, with probability at least $1 - e^{-4m}$, there exist $a \in \{0,\ldots,\log\eps^{-1}\}$, $k \in \{1,\ldots,2m\}\setminus K$ and $v \in S_{a,k}$ for which $d(v,\{u_1,\ldots,u_\ell\}) > 2^{-a-1}$. Consider the valid composition:
			$\tilde{A} = (u_1,a_2,u_2,\ldots,a_{i_0},u_{i_0},a,v)$, where $i_0 = \max(\{1\}\cup \{i | a_i \ge a\})$ (possibly $i_0 = 1$, and in this case $\tilde{A} = (u_1,a,v)$). Comparing the ranks,
			\[
			\vec{r}(A)
			= (a_2,\ldots,a_{i_0},\mathbf{a_{i_0+1}},\ldots)
			< (a_2,\ldots,a_{i_0},\mathbf{a})
			= \vec{r}(\tilde{A})
			\]
			Note that possibly also $i_0 = \ell$, and in this case $a_{i_0+1}$ is the virtual ``end-of-string'' character that is defined as smaller than every value of $\mathbf{a}$. Hence in all cases $\vec{r}(\tilde{A}) > \vec{r}(A)$ as desired.
		\end{proof}
		
		\begin{theorem} \label{th:alg-m-supp-na-correct}
			Algorithm \ref{alg:m-supp-na} is a one-sided $\eps$-test of $\mathcal{S}_m$ that makes $O(\eps^{-1} \log \eps^{-1} \cdot m \log m)$ queries.
		\end{theorem}
		\begin{proof}
			For the query complexity, note that for every $1 \le k \le 2m$ and for every $0 \le a \le \log\eps^{-1}$, the algorithm makes $\ceil{2^{3-a}\eps^{-1}} \cdot \ceil{2^{a+2} \log (m+1)} = O(\eps^{-1} \log m)$ queries inside $S_{a,k}$. Since the number of pairs of $(a,k)$ is $O(m \log\eps^{-1})$, the query complexity is $O(\eps^{-1} \log\eps^{-1} \cdot m \log m)$.
			
			Perfect completeness is trivial, since a distinguishable composition of length $m+1$ is in particular an explicit witness against $\mathcal{S}_m$. For soundness, recall Lemma \ref{lemma:maximum-must-be-long}. If none of the bad events $B_{K,\ell}$ happens, then there exists a valid composition $A$ of size $m+1$. If the input is $\eps$-far from $\mathcal{S}_m$, then by the union bound (of the complement events), the probability for this is at least
			\[1 - \sum_{\ell=1}^{m} \cdot \binom{2m}{\ell} \cdot e^{-4m} \ge 1 - 2^{2m} e^{-4m} > \frac{99}{100} \mathrm{\ for\ } m \ge 2 \]
			Hence, with probability $\frac{99}{100}$, there exists a valid composition $A$ of size $m+1$.
			
			Assume that this happens with some valid composition $A = (u_1,a_2,u_2,\ldots,a_{m+1},u_{m+1})$. Set $a_1 = \log\eps^{-1}$ for convenience.
			
			By the constraints of $A$, for every $1 \le i < j \le m+1$:
			\begin{itemize}
				\item $d(u_i,u_j) > 2^{-a_j-1}$.
				\item $u_i$ is queried in $I_{a_i} \supseteq I_{a_j}$.
				\item $u_j$ is queried in $I_{a_j}$.
			\end{itemize}
			The probability that $I_{a_j}$ distinguishes $u_i$ and $u_j$ is at least $1 - (1 - 2^{-a_j-1})^{\ceil{2^{a_j + 2}\log (m+1)}} > 1 - e^{-2 \log (m+1)}$. The probability that this happens for every $1 \le i < j \le m + 1$ is bounded by $1 - \binom{m + 1}{2}e^{-2 \log (m+1)}$, which is at least $\frac{5}{6}$ (for $m \ge 2$). Overall, with probability at least $\frac{99}{100} \cdot \frac{5}{6} > \frac{2}{3}$, there exists a valid composition $A$ of size $m+1$ which is also distinguishable, and in this case the algorithm rejects.
		\end{proof}

		\section{The fishing expedition paradigm}
		\label{sec:fishing-expedition}
		
		%\ todo{\sout{If splitted half is moved to prelims then fine, otherwise it might be remerged} \sout{Now move it to before current Section 9}}
		
		We construct and prove here the algorithm for the fishing expedition paradigm, Lemma \ref{lemma:fishing-expedition-full}.
		
		%\begin{definition}[Fail stability] \label{def:fail-stability}
		%    Let $\mathcal{A}$ be a subroutine that may succeed with some output, or fail (with no additional output). Specifically let $R_1,\ldots,R_N$ be random variables that detail the outputs of the first $N$ executions of $\mathcal{A}$, where a $0$ value means that an execution has failed (and all other values imply success). We say that $\mathcal{A}$ \emph{is fail stable}, if for every $2 \le i \le N$ and every result sequence $(r_1,\ldots,r_{i-2})\in\supp(R_1,\ldots,R_{i-2})$:
		%    \[\Pr\left[R_i \ne 0 \mid R_1 = r_1, \ldots, R_{i-2} = r_{i-2}, R_{i-1} = 0 \right] = \Pr\left[R_{i-1} \ne 0 \mid R_1 = r_1, \ldots, R_{i-2} = r_{i-2} \right]\]
		%    In other words, a failure does not affect the probability of further executions to succeed.
		%\end{definition}
		
		%\begin{definition}[Diminishing returns] \label{def:diminishing-returns}
		%    Let $\mathcal{A}$ and $R_1,\ldots,R_N$ be as in Definition \ref{def:fail-stability}. We say that $\mathcal{A}$ \emph{has diminishing returns}, if for every $2 \le i \le N$ and every result sequence $(r_1,\ldots,r_{i-1})\in\supp(R_1,\ldots,R_{i-1})$:
		%    \[\Pr\left[R_i \ne 0 \mid R_1 = r_1, \ldots, R_{i-2} = r_{i-2}, R_{i-1} = r_{i-1} \right] \le \Pr\left[R_{i-1} \ne 0 \mid R_1 = r_1, \ldots, R_{i-2} = r_{i-2}\right]\]
		%    That is, if $\mathcal{A}$ has diminishing returns, then a success in a single execution never increases, but may decrease, the probability of further executions to succeed.
		%\end{definition}
		
		The algorithm has three parameters: a threshold $p$, a confidence $q$ and a goal $k \ge 1$. The input is a subroutine $\mathcal{A}$ with diminishing returns and fail stability. Informally, the goal of the algorithm is to have $k$ successful executions of $\mathcal{A}$, but also to terminate earlier if the probability of $\mathcal{A}$ to succeed becomes lower than $p$. Since the algorithm has no actual access to the success probability of $\mathcal{A}$, it should terminate early only if it is confident enough that the success probability of further executions is too low for them to be effective.
		
		%We consider a black box subroutine $\mathcal{A}$ with fail stability (Definition \ref{def:fail-stability}) and diminishing returns (Definition \ref{def:diminishing-returns}). Recall the fishing expedition lemma:
		%\begin{lemma}
		%    \label{lemma:fishing-expedition-full}
		%    For an arbitrary algorithm that repeatedly executes $\mathcal{A}$ as a black box, we define the following random variables:
		%    \begin{itemize}
			%        \item $N$ -- the number of executions.
			%        \item $R_1,\ldots,R_N$ -- their outcomes.
			%        \item $X_1,\ldots,X_N$ -- indicators of success (that is, $X_i = 1$ if and only if $R_i \ne 0$).
			%        \item $H = \sum_{i=1}^N X_i$ -- the number of successful executions.
			%        \item $\hat p = \Pr[X_{N+1} = 1 | R_1,\ldots,R_N]$ -- the success probability of an additional, hypothetical execution of $\mathcal{A}$.
			%    \end{itemize}
		%
		%    Let $p > 0$, $q > 0$, $k \ge 1$, and let $\mathcal{A}$ be a subroutine with diminishing returns and fail stability. There exists such algorithm for which $N \le p^{-1}(4H + 5(\log q^{-1} + \log (\log k + 1))) + 1$ and $H \le k$, and with probability higher than $1-q$, either $H = k$ or $\hat p \le p$ (or both).
		%\end{lemma}
		%\lemmaZfishingZexpeditionZfull*
		%In this section we prove Lemma \ref{lemma:fishing-expedition-full}. The algorithm guaranteed by the lemma is Algorithm \ref{alg:fishing-expedition}.
		In this section we construct the fishing expedition paradigm providing Lemma \ref{lemma:fishing-expedition-full}. The algorithm providing it is Algorithm \ref{alg:fishing-expedition}. The observations and lemmas below show that this algorithm satisfies the corresponding components of Lemma \ref{lemma:fishing-expedition-full}.
		
		Algorithm \ref{alg:fishing-expedition} repeatedly executes $\mathcal{A}$. Of course, if $\mathcal{A}$ was successful for the $k$-th time, the algorithm terminates immediately. At some predefined check points (which are determined by $p$, $q$ and $k$), the algorithm considers an early termination. Concretely, for $t_\mathrm{max} = \floor{\log k + 1}$ and for every $2 \le t \le t_\mathrm{max}$, after $\ceil{p^{-1}\cdot \max\{2^t, 5(\log q^{-1} + \log (\log k + 2))\}}$ executions of $\mathcal{A}$, the algorithm terminates if the number of successful executions was less than a $\frac{1}{2}p$-portion of the total number of executions. The algorithm must terminate in one of these iterations, as stated in the following observation.
		
		\begin{algorithm}
			\caption{Fishing expedition}
			\label{alg:fishing-expedition}
			\begin{algorithmic}
				\State \textbf{parameters} $k \ge 1$ (goal), $p > 0$ (threshold), $q > 0$ (confidence).
				\State \textbf{input} A subroutine $\mathcal{A}$ with output, given as a black box, where the output ``$0$'' means \textsc{fail}.
				
				\State \textbf{let} $t_\mathrm{max} \gets \floor{\log k + 1}$.
				
				\State \textbf{let} $N_1 \gets 0$.
				
				\State \textbf{set} $H \gets 0$.
				
				\For{$t$ \textbf{from} $2$ \textbf{to} $t_\mathrm{max}$}
				\State \textbf{let} $N_t \gets \ceil{p^{-1}\max\{2^t, 5(\log q^{-1} + \log (\log k + 1))\}}$.
				
				\For{$N$ \textbf{from} $N_{t-1} + 1$ \textbf{to} $N_t$ \Comment{possibly empty}}
				\State \textbf{run} $\mathcal{A}$, let $R_N$ be its outcome.
				
				\State \textbf{let} $X_N$ be an indicator for success ($X_N = 1$ if $R_N \in G$, otherwise $X_N = 0$).
				
				\State \textbf{set} $H \gets H + X_N$.
				
				\If{$H = k$}
				\textbf{terminate} with $N$. \Comment{goal is reached}
				\EndIf
				\EndFor
				\If{$H < \frac{1}{2}pN_t$}
				\State \textbf{terminate} with $N_t$. \Comment{continuing is ineffective}
				\EndIf
				\EndFor
				\State \textbf{unreachable point} \Comment{Observation \ref{obs:fishing-expedition-unreachable-point}.}
			\end{algorithmic}
			
		\end{algorithm}
		
		\begin{observation} \label{obs:fishing-expedition-unreachable-point}
			If Algorithm \ref{alg:fishing-expedition} has not terminated by the $t_\mathrm{max}$-th iteration, it must do so there.
		\end{observation}
		\begin{proof}
			Note that $N_{t_\mathrm{max}} = N_{\floor{\log k + 1}} \ge \ceil{p^{-1} 2^{\floor{\log k + 2}}} \ge p^{-1} 2^{\log k + 1} = 2p^{-1}k$. That is, at the end of the $t_\mathrm{max}$-th iteration, if $H < k$ then $H < \frac{1}{2}pN_{t_\mathrm{max}}$ and the iteration must terminate.
		\end{proof}
		
		For every $N \ge 0$, let $H_N$ be the value of $H$ after the $N$-th execution (that is, $H_N = \sum_{i=1}^N X_i$).
		
		\begin{lemma} \label{lemma:fishing-expedition-first-constraint}
			Algorithm \ref{alg:fishing-expedition} always terminates with $N \le p^{-1}(4H + 5(\log q^{-1} + \log (\log k + 1))) + 1$.
		\end{lemma}
		\begin{proof}
			If the algorithm terminates in the first iteration ($t=2$), then
			\[N \le N_2 = \ceil{p^{-1} \max\{2^2, 5(\log q^{-1} + \log (\log k + 2))\}} \le 5p^{-1}(\log q^{-1} + \log(\log k + 2)) + 1 \]
			
			In terminations outside the first iteration, observe that $N_t \le 2N_{t-1}$. Since the algorithm did not terminate in the previous iteration, $H_{N_{t-1}} \ge \frac{1}{2}pN_{t-1} \ge \frac{1}{4}pN_t$. Since $H_{N_t} \ge H_{N_{t-1}}$ as well, we have $H_{N_t} \ge \frac{1}{4}pN_t$ and thus $N_t \le 4p^{-1}H_{N_t}$.
			
			By Observation \ref{obs:fishing-expedition-unreachable-point}, the algorithm must have terminated in one of the iterations. This completes the proof.
		\end{proof}
		
		\begin{lemma} \label{lemma:fishing-expedition-second-constraint}
			Let $p > 0$, $q > 0$, $k \ge 1$, and let $\mathcal{A}$ be a subroutine with diminishing returns and fail stability. Let $N$ be the number of executions of $\mathcal{A}$ done by Algorithm \ref{alg:fishing-expedition} and let $H$ be the number of successful executions. Let $\hat{p}=\Pr[X_{N_t+1}=1 | R_1,\ldots,R_{N_t} ]$ be the probability that an additional, hypothetical execution of $\mathcal{A}$ is successful (note that $\hat{p}$ is a random variable that depends on the outcomes of the $N$ executions of $\mathcal{A}$).
			In this setting, with probability higher than $1-q$, $H = k$ or $\hat{p} \le p$ (or both).
		\end{lemma}
		
		\begin{proof}
			Consider the following equivalent algorithm: we simulate Algorithm \ref{alg:fishing-expedition}, but ignore the termination requests. When reaching what was the ``unreachable point'', we stop and choose the output $(N,H)$ according to the first termination request. Clearly, this algorithm always returns the same $N$, $H$ as would be in a run of Algorithm \ref{alg:fishing-expedition} for the same outcome sequence, but it is non-adaptive (more precisely, it always makes $N_{t_\mathrm{max}}$ executions of $\mathcal{A}$ regardless of their outcomes and then chooses the output).
			
			For every $0 \le n \le N_{t_\mathrm{max}}$, let $H_n = \sum_{i=1}^n X_i$. For every $2 \le t \le t_\mathrm{max}$, let $\mathcal{B}_t$ be the following bad event: $(H_{N_t} < k) \wedge (H_{N_t} < \frac{1}{2}pN_t) \wedge (\Pr[X_{N+1} = 1 | R_1,\ldots,R_{N}] \ge p)$. Note that if no bad event happens, then the output satisfies $H \ge k$ or $\Pr[X_{N_t+1} = 1 | R_1,\ldots,R_{N_t}] < p$. We will use a variant of Chernoff's bound that is proved in the appendix (Lemma \ref{lemma:chernoff-with-goal}) to bound the probabilities of the bad events.
			
			Let $\mathcal{G} = \left\{ (R_1,\ldots,R_N) \in \supp(R_1,\ldots,R_N) : \left|\{R_i \ne 0\}\right| \ge k \vee \Pr[R_{N+1} \ne 0 | R_1,\ldots,R_N] < p \right\}$ be the set of outcome sequences where a termination would be justifiable (note that it is not exactly the same as the set of termination conditions). Note that if $(R_1,\ldots,R_{N-1}) \in \mathcal{G}$, then $(R_1,\ldots,R_N) \in \mathcal{G}$ as well, since the number of non-zero elements cannot decrease and the probability of the next trial cannot increase.
			
			For every $2 \le t \le t_\mathrm{max}$, by Lemma \ref{lemma:chernoff-with-goal}, for $\delta=\frac{1}{2}$, $m = N_t$ and $X = H_{N_t}$
			\begin{eqnarray*}
				\Pr[\mathcal{B}_t]
				&=& \Pr\left[(H_{N_t} < k) \wedge (H_{N_t} < \frac{1}{2}pN_t) \wedge (\Pr\left[R_{N_t+1} \notin G \cond R_1,\ldots,R_{N_t})\right] \ge p\right] \\
				&=& \Pr\left[((R_1,\ldots,R_{N_t}) \notin \mathcal{G}) \wedge (H_{N_t} < \frac{1}{2}pN_t)\right] \\
				&<& (\sqrt{2/e})^{pN_t}
				\le 0.86^{pN_t}
				\le 0.86^{5(\log q^{-1} + \log (\log k + 2))}
				< 0.5^{\log q^{-1} + \log(\log k + 2)}
				= \frac{q}{\log k + 2}
			\end{eqnarray*}
			
			By the union bound, $\Pr[\bigvee_{t=2}^{t_\mathrm{max}} \mathcal{B}_t] < \floor{\log k + 2} \cdot \frac{q}{\log k + 2} \le q$. With probability greater than $1-q$, no bad event happens, and the algorithm terminates with $H \ge k$ or $\Pr\left[R_{N+1} \notin G \cond R_1,\ldots,R_N\right] < p$.
		\end{proof}
		
		At this point we can prove the fishing expedition lemma.
		\begin{proof}[Proof of Lemma \ref{lemma:fishing-expedition-full}]
			This lemma follows immediately from Lemma \ref{lemma:fishing-expedition-first-constraint} (number of executions) and Lemma \ref{lemma:fishing-expedition-second-constraint} (the algorithm reaches one of its goals with probability at least $1-q$).
		\end{proof}
		
		\section{Adaptive $m$-support test}
		
		We construct here an adaptive one-sided error algorithm using $O(\eps^{-1} m \log m \cdot \min\{\log \eps^{-1}, \log m\})$ many queries.
		
		\paragraph{The advantage of being adaptive}
		
		The non-adaptive algorithm considers $\Omega(\log \eps^{-1})$ buckets of distance ranges at the cost of $\Omega(\eps^{-1} \cdot m \log m)$ queries per bucket, and we believe that we cannot do much better (the $\Omega(\log \eps^{-1})$ buckets are required according to the non-adaptive lower bound, and the $\Omega(\eps^{-1} m \log m)$ queries per bucket are required according to the adaptive lower-bound in Section \ref{sec:lbnd-adaptive-one-sided}). For $\eps < \frac{1}{m}$, the number of distance buckets is more than $\log m + 1$.
		
		An adaptive algorithm can do better. Initially, we consider $\log m + 1$ distance buckets. Then, using a decision tree constructed incrementally as more distinct elements are found, we avoid the need to consider the rest of the buckets. In particular, the ``far buckets'' phase is the bottleneck of the algorithm ($\Theta(m^3 \log m + \eps^{-1} m \log^2 m)$ queries), where the second phase is extremely cheap: $O(\eps^{-1} m)$ queries, below the lower bound for adaptive algorithms (see Section \ref{sec:lbnd-adaptive-one-sided}). This means that further improvements of the query complexity must address the first phase (which only considers distinct elements that are $\frac{1}{2m}$-far from each other). We later get rid of the $\Theta(m^3\log m)$ term by using the non-adaptive algorithm when $\eps$ is too large.
		
		\subsection{Additional building blocks}
		
		The ``fishing expedition'' paradigm (Algorithm \ref{alg:fishing-expedition}) is an important building block in our algorithm. Here we define some additional algorithmic building blocks.
		
		\begin{building-block}[Using a decision tree] \label{bb:using-a-decision-tree}
			Let $A = \{x_1,\ldots,x_k\}$ be a set of $k$ distinct strings, and let $\mathcal{T}$ be a query-based decision tree with exactly $k$ leaves, such that every string in $A$ corresponds to a different leaf in $\mathcal{T}$. For every input string $x$, we can algorithmically find an $1 \le i \le k$ such that $\mathcal{T}(x) = \mathcal{T}(x_i)$. The number of queries made by this procedure is bounded by the height of $\mathcal{T}$.
		\end{building-block}
		\begin{proof}
			Trivially, start at the root and follow the path according to the queries and the answers of $x$ to these queries. At some point we reach a leaf. This leaf must correspond to some $x_i \in A$, since $\mathcal{T}$ has exactly $k$ leaves and they fully distinguish the $k$ elements in $A$.
		\end{proof}
		
		\begin{building-block}[Updating a decision tree] \label{bb:insertion-to-decision-tree}
			Let $\mathcal{T}$ be a decision tree with $k$ leaves that fully distinguishes $A = \{x_1,\ldots,x_k\}$. Given a string $x$, an $i$ for which $\mathcal{T}(x) = \mathcal{T}(x_i)$, and an index $j$ for which $x|_j \ne x_i|_j$, we can algorithmically construct a decision tree $\mathcal{T}'$ with $k+1$ leaves that fully distinguishes $A' = \{x_1,\ldots,x_k,x\}$, at the cost of no additional queries.
		\end{building-block}
		\begin{proof}
			Just substitute the leaf $\mathcal{T}(x_i)$ by an internal node consisting of the query $j$ and two children $x$ and $x_i$.
		\end{proof}
		
		\begin{building-block}[Construction of a decision tree]
			\label{bb:construction-of-a-decision-tree}
			Let $A = \{x_1,\ldots,x_k\}$ be a distinguishable set of strings, that is, for every $1 \le i_1 < i_2 \le k$, there exists an index $j$ at which both $x_{i_1}$, $x_{i_2}$ were queried and $(x_{i_1})_j \ne (x_{i_2})_j$. We can construct a decision tree $\mathcal{T}$ with exactly $k$ leaves that distinguishes all $x_1,\ldots,x_k$, at the cost of at most $k^2$ queries. Moreover, at the end of the construction, for every $1 \le i \le k$, every $x_i$ was queried at all indices in the search path of $x_i$ in $\mathcal{T}$.
		\end{building-block}
		\begin{proof}
			For $k=1$ it is trivial. For $k>1$, consider $A' = \{x_1,\ldots,x_{k-1}\}$ and construct a decision tree $\mathcal{T}'$ with $k-1$ leaves distinguishing $x_1,\ldots,x_{k-1}$ at the cost of at most $(k-1)^2$ queries. Use Building Block \ref{bb:using-a-decision-tree} (using a decision tree) to find $1 \le i \le k-1$ for which $\mathcal{T}'(x_k) = \mathcal{T}'(x_i)$, at the cost of at most $k-1$ queries (note that these additional queries are only done in $x_k$). According to the statement, there exists some index $j$ at which both $x_i$ and $x_k$ were queried, and $(x_i)_j \ne (x_k)_j$. Use Building Block \ref{bb:insertion-to-decision-tree} (updating a decision tree) to insert $x_k$ to the tree at the cost of no additional queries.
			
			We used at most $(k-1)^2$ queries to construct $\mathcal{T}'$ and at most $k-1$ additional queries to insert $x_k$ to it. The total number of queries is at most $k^2$, as required.
		\end{proof}
		
		\subsection{The algorithm}
		
		If $\eps \ge \frac{1}{m^2}$ then we just execute Algorithm \ref{alg:m-supp-na}. The query complexity of the algorithm is $O(\eps^{-1} \log \eps^{-1} \cdot m \log m)$ in this case, and it is the same as $O(\eps^{-1} m \log m \cdot \min\{\log \eps^{-1}, \log m\})$ since $\log \eps^{-1} \le \log m^2 \le 2 \log m = O(\log m)$. If $\eps < \frac{1}{m^2}$, then we use the adaptive algorithm below.
		
		\begin{subalgorithms}
			
			For $\eps < \frac{1}{m^2}$, the algorithm consists of two phases: the first one is intended to find distinct samples that are $\frac{1}{2m}$-far from each other, and the second one uses a decision tree to reduce the number of queries required to find additional distinct elements that are $\frac{1}{2m}$-close to those already found.
			
			\paragraph{Batches} Assume that we know (or guess) that $\Pr[ d(x,A) > \alpha] > \frac{\alpha^{-1}\eps}{\log m}$. If we draw $\alpha \eps^{-1} \log m$ samples, then with high probability there is a sample $Y$ that is $\alpha$-far from $A$. In this case, a set of $2^a \log m$ indices should distinguish $Y$ from all $X \in A$. That is, under this assumption, we can find an additional distinct element with probability greater than some global positive constant. This subroutine is called a \emph{batch}.
			
			Concretely, every $a$-batch chooses a set $J$ of $\ceil{2^{a+2} \log m}$ indices (uniformly and independently) and queries all samples in $A$ at the indices of $J$. Then it draws additional $\ceil{2^{2-a} \eps^{-1} \log m}$ samples and queries all of them at the indices of $J$. If there exists a sample $Y$ for which $Y|_J \ne X|_J$ for every $X \in A$, the batch is considered successful, and we add $Y$ to $A$.
			
			\begin{algorithm}
				\caption{Adaptive one-sided $\eps$-test for $\mathcal{S}_m$, a single batch}
				\label{alg:m-supp-adaptive-first-phase-batch}
				\begin{algorithmic}
					\State \textbf{parameters} $\eps > 0$, $A$, $m \ge 2$, $0 \le a \le \ceil{\log m}$ where $|A| \le m$.
					\State \textbf{input} A distribution $P$.
					\State \textbf{choose} a set $J$ of $\ceil{2^{a+2} \log m}$ indices uniformly and independently.
					\State \textbf{query} $X$ at $J$ for every $X \in A$.
					\State \textbf{take} $\ceil{2^{2-a} \eps^{-1} \log m}$ samples.
					\State \textbf{query} each new sample at $J$.
					\If{there exists a sample $Y$ for which $Y|_J \ne X_J$ for every $X \in A$}
					\State \textbf{set} $A \gets A \cup \{Y\}$.
					\State \Return \textsc{success} with $(Y,J)$.
					\Else
					\State \Return \textsc{fail}
					\EndIf
				\end{algorithmic}
			\end{algorithm}
			
			\begin{observation}
				Algorithm \ref{alg:m-supp-adaptive-first-phase-batch} has diminishing returns and fail stability as per Definition \ref{def:fail-stability} and Definition \ref{def:diminishing-returns}, where for formality's sake we use some fixed mapping of the set of possible non-failing output values to distinct positive natural numbers.
			\end{observation}
			
			\begin{lemma} \label{lemma:alg-m-supp-adaptive-first-phase-batch-correct}
				Algorithm \ref{alg:m-supp-adaptive-first-phase-batch} uses $O(m^2 \log m + \eps^{-1} \log^2 m)$ queries, and if $\Pr[d(x,A) > 2^{-a-1}] > \frac{2^a \eps}{4 \log m}$, then the success probability of Algorithm \ref{alg:m-supp-adaptive-first-phase-batch} is at least $\frac{1}{3}$.
			\end{lemma}
			\begin{proof}
				The query complexity of Algorithm \ref{alg:m-supp-adaptive-first-phase-batch} is $(|A| + \ceil{2^{-a} \eps^{-1} \log m})|J| \le (m + 2^{-a}\eps^{-1}\log m + 1)(2^a \log m + 1) = O(m^2 \log m + \eps^{-1} \log^2 m)$ (since $a \le \log m + 1$).
				
				If $\Pr[d(x,A) > 2^{-a-1}] > \frac{2^a \eps}{4 \log m}$, then the expected number of samples $2^{-a-1}$-far from $A$ within the new $\ceil{2^{2-a} \eps^{-1} \log m}$ samples is at least $1$. Hence the probability that there is at least one them is at least $1 - e^{-1} > \frac{3}{5}$ (by Lemma \ref{lemma:pr-chernoff-zero}).
				
				If this happens, let $Y$ be such a sample. With probability at least $1 - m(1 - 2^{-a-1})^{2^{a+2} \log m} > \frac{7}{10}$, $Y|_J \ne X|_J$ for every $X \in A$. Overall, the probability to have a sample $Y$ for which $Y|_J \ne X|_J$ for all $X \in A$ is at least $\frac{3}{5} \cdot \frac{7}{10} > \frac{1}{3}$.
			\end{proof}
			
			\paragraph{The first phase} The batches are not standalone, since they must have some parameter $a$. The first phase of the algorithm consists of $\ceil{\log m} + 1$ iterations. For every $0 \le a \le \ceil{\log m}$, the $a$-th iteration consists of batches with parameter $a$. To make sure that the batches are only executed when it is cost-effective, we use the ``fishing expedition'' paradigm (Algorithm \ref{alg:fishing-expedition}).
			
			\begin{algorithm}
				\caption{Adaptive one-sided $\eps$-test for $\mathcal{S}_m$, first phase}
				\label{alg:m-supp-adaptive-first-phase}
				\begin{algorithmic}
					\State \textbf{parameters} $\eps > 0$, $m \ge 2$.
					
					\State \textbf{input} A distribution $P$, a set $A \subseteq \supp(P)$ of distinguishable elements.
					
					\For{$a$ \textbf{from} $0$ \textbf{to} $\ceil{\log m}$}
					\State \textbf{let} $k_a = m + 1 - |A|$.
					
					\State \textbf{run} Algorithm \ref{alg:fishing-expedition} (``fishing expedition'') with parameters $k = k_a$, $q=\frac{1}{4\ceil{\log m + 1}}$, $p = \frac{1}{3}$,
					
					\State \phantom{\textbf{run} Algorithm \ref{alg:fishing-expedition} (``fishing expedition'') with} input $\mathcal{A} = $ Algorithm \ref{alg:m-supp-adaptive-first-phase-batch} (a single batch).
					
					\If{$|A| \ge m + 1$}
					\State \Return \reject
					\EndIf
					\EndFor
					
					\State Proceed to the second phase with $A$.
				\end{algorithmic}
			\end{algorithm}
			
			\begin{lemma} \label{lemma:alg-m-supp-adaptive-first-phase-correct}
				Algorithm \ref{alg:m-supp-adaptive-first-phase} makes $O(m^3 \log m + \eps^{-1} m \log^2 m)$ queries. With probability at least $\frac{3}{4}$, either $|A| \ge m + 1$ or $\Ct[d(x,A) | d(x,A) > \frac{1}{2m}] \le \frac{1}{2}\eps$.
			\end{lemma}
			\begin{proof}
				For every $0 \le a \le \ceil{\log m}$, let $H_a$ be the number of successful executions of Algorithm \ref{alg:m-supp-adaptive-first-phase-batch} within the $a$-th iteration. Also, let $N_a$ be the total number of executions. Lemma \ref{lemma:fishing-expedition-first-constraint} guarantees that
				\[N_a \le 3\left(4 H_a + 5\left(\log \left(4 \ceil{\log m + 1}\right) + \log\left(\log m + 2\right)\right)\right) + 1 \le 12H_a + O(\log m)\]
				
				Note that $\sum_{a=0}^{\ceil{\log m}} H_a = |A| - 1$. Either $|A| \ge m + 1$ or $\sum_{a=0}^{\ceil{\log m}} H_a < m$. Considering all iterations,
				\begin{eqnarray*}
					\sum_{a=0}^{\ceil{\log m}} N_a
					&\le& 12 \sum_{a=0}^{\ceil{\log m}} H_a + (\ceil{\log m} + 1)O(\log m) \\
					&<& 12m + O(\log^2 m) = O(m)
				\end{eqnarray*}
				
				Every iteration makes at most $O(m^2 \log m + \eps^{-1} \log^2 m)$ queries, hence the total number of queries is $O(m^3 \log m + \eps^{-1} m \log^2 m)$.
				
				For every $0 \le a \le \ceil{\log m}$, by Lemma \ref{lemma:fishing-expedition-second-constraint}, with probability at least $1 - \frac{1}{4\ceil{\log m}}$, either $H_a = k_a$ or the probability to find an additional distinct element is less than $\frac{1}{3}$. In the first case, $|A| = m + 1$ at the end of the iteration (since $k_a = m + 1 - |A|$, considering the size of $A$ at the beginning of the iteration). In the second case, $\Pr[d(x,A) > 2^{-a-1}] \le \frac{2^a \eps}{4 \log m}$ (by Lemma \ref{lemma:alg-m-supp-adaptive-first-phase-batch-correct}).
				
				With probability at least $1 - \frac{\ceil{\log m} + 1}{4 \ceil{\log m + 1}} = \frac{3}{4}$, this happens for all values of $a$. That is, either $|A| \ge m + 1$ (once) or $\Pr[d(x,A) > 2^{-a-1}] \le \frac{2^a \eps}{4 \log m}$ for all $0 \le a \le \ceil{\log m}$. In the first case we reject, and in the second case,
				\begin{eqnarray*}
					\Ct\left[d(x,A) | d(x,A) > \frac{1}{2m}\right]
					&\le& \sum_{a=0}^{\ceil{\log m}} 2^{-a} \Pr[2^{-a-1} < d(x,A) \le 2^{-a}] \\
					&\le& \sum_{a=0}^{\ceil{\log m}} 2^{-a} \Pr[d(x,A) > 2^{-a-1}] \\
					&\le& \sum_{a=0}^{\ceil{\log m}} 2^{-a}\cdot \frac{2^a \eps}{4 \log m}
					= \frac{\ceil{\log m} + 1}{4 \log m}\eps \le \frac{1}{2}\eps
			\end{eqnarray*}\end{proof}
			
			\paragraph{The second phase}
			The second phase of the algorithm handles the case where $|A| \le m$ and also $\Ct[d(x,A) | d(x,A) \ge \frac{1}{2m}] \le \frac{1}{2}\eps$. If the input distribution is $\eps$-far from being supported by any subset of $A$, the contribution of ``small distances'', $\Ct[d(x,A) | d(x,A) < \frac{1}{2m}]$, is strictly greater than $\frac{1}{2}\eps$.
			
			First, we construct a decision tree $\mathcal{T}$ for the elements in $A$, according to Building Block \ref{bb:construction-of-a-decision-tree}, at the cost of at most $m^2$ queries. After the decision tree is constructed, the algorithm is iterative. It tracks a set $A = \{X_0,\ldots,X_{|A|-1}\}$ of distinguishable samples (initialized with the $A$ supplied by the first phase) and a decision tree $\mathcal{T}$ with $|A|$ leaves corresponding to $A$'s elements.
			
			In every iteration, the algorithm draws a new sample $Y \sim P$ and executes the decision tree $\mathcal{T}$ on $Y$, resulting in an index $0 \le i \le |A| - 1$ for which $\mathcal{T}(Y) = \mathcal{T}(X_i)$. Then it queries both $X_i$ and $Y$ at a brand new query set $J$ of size $m$. If $Y$ is $\frac{1}{2m}$-close to $A$, then with high probability (proportionally to the distance), $Y|_J \ne X_i|_J$. If this happens, then we add $Y$ to $A$ and to the decision tree.
			
			\begin{algorithm}
				\caption{Adaptive one-sided $\eps$-test for $\mathcal{S}_m$, a single iteration of the second phase}
				\label{alg:m-supp-adaptive-second-phase-iteration}
				\begin{algorithmic}
					\State \textbf{input} A sample $Y \in \supp(P)$, $A \subseteq \supp(P)$, a decision tree $\mathcal{T}$; $|A| \ge 1$.
					\State \textbf{invariant} $\mathcal{T}$ has $|A|$ leaves corresponding to $A$'s elements.
					
					\State \textbf{choose} a set $J$ of $m$ indices uniformly, independently, with repetitions.
					
					\State \textbf{let} $X \in A$ for which $\mathcal{T}(Y) = \mathcal{T}(X)$. \Comment{Building Block \ref{bb:using-a-decision-tree}}
					
					\State \textbf{query} $X, Y$ at $J$.
					
					\If{$Y|_J \ne X|_J$}
					\State \textbf{set} $A \gets A \cup \{Y\}$.
					
					\State \textbf{add} $Y$ to $\mathcal{T}$. \Comment{Building Block \ref{bb:insertion-to-decision-tree}, no extra queries}
					\EndIf
				\end{algorithmic}
			\end{algorithm}
			
			\begin{lemma} \label{lemma:m-supp-adaptive-second-phase-iteration-correct}
				Algorithm \ref{alg:m-supp-adaptive-second-phase-iteration} makes at most $3m$ queries and keeps its invariants. If $|A| \le m$ and the input distribution $P$ is $\eps$-far from $\mathcal{S}_m$, then the algorithm adds $Y$ to $A$ with probability at least $\frac{1}{4}m\eps$.
			\end{lemma}
			\begin{proof}
				There exists some $\hat{A} \subseteq A$ of size at most $m$ for which $\Ct[d(x,\hat{A}) | d(x,\hat{A}) > \frac{1}{2m}] \le \frac{1}{2}\eps$ ($\hat{A}$ is the output of the first, non-adaptive phase). Since $A$ contains $\hat{A}$, $\Ct[d(x,A) | d(x,A) > \frac{1}{2m}] \le \frac{1}{2}\eps$ as well. Since $|A| \le m$, $\E[d(x,A)] \ge d(P,\mathcal{S}_m) > \eps$, and $\Ct[d(x,A) | d(x,A) \le \frac{1}{2m}] > \frac{1}{2}\eps$.
				
				For every $Y\in\supp(P)$, consider $X(Y) \in A$ for which $\mathcal{T}(Y) = \mathcal{T}(X(Y))$. For a sample $Y$ drawn from $P$, by Lemma \ref{lemma:mapping-contribution}, The probability to distinguish between $X(Y)$ and $Y$ is:
				\begin{eqnarray*}
					\E_{Y\sim P}\left[1 - \left(1 - d\left(Y,X(Y)\right)\right)^{|J|}\right]
					&\ge& \Ct_{Y\sim P}\left[1 - \left(1 - d\left(Y,X(Y)\right)\right)^{|J|} \cond d(Y,A) \le \frac{1}{2m}\right] \\
					&\ge& \Ct_{Y\sim P}\left[1 - \left(1 - d(Y,A)\right)^m \cond d(Y,A) \le \frac{1}{2m}\right] \\
					&\ge& \frac{1}{2}m \Ct_{Y\sim P}\left[d(Y,A) \cond d(Y,A) \le \frac{1}{2m}\right]
					> \frac{1}{4}m\eps
				\end{eqnarray*}
				If the algorithm distinguishes between $Y$ and $X(Y)$, then the invariant is kept by the constraints of Building Block \ref{bb:insertion-to-decision-tree}. If the algorithm fails to distinguish between $Y$ and $X(Y)$, then the invariant is kept trivially.
			\end{proof}
			
			\begin{algorithm}
				\caption{Adaptive one-sided $\eps$-test for $\mathcal{S}_m$}
				\label{alg:m-supp-adaptive}
				\begin{algorithmic}
					\State \textbf{input} A distribution $P$.
					
					\If{$\eps \ge \frac{1}{m^2}$}
					\State \textbf{run} Algorithm \ref{alg:m-supp-na} and \Return its answer.
					\EndIf
					
					\State \textbf{take} the first sample $u$.
					
					\State \textbf{set} $A \gets \{u\}$.
					
					\State \textbf{run} Algorithm \ref{alg:m-supp-adaptive-first-phase} (possibly modifying $A$, possibly rejecting).
					
					\State \textbf{construct} a decision tree $\mathcal{T}$ based on $A$. \Comment{Building Block \ref{bb:construction-of-a-decision-tree}}
					
					\State \textbf{invariant}$\mathcal{T}$ has $|A|$ leaves corresponding to $A$'s elements.
					
					\For{$\ceil{48 \eps^{-1}}$ \textbf{times}}
					\State \textbf{draw} another sample $Y$.
					
					\State \textbf{run} Algorithm \ref{alg:m-supp-adaptive-second-phase-iteration} with $(Y, A, \mathcal{T})$ (note that $A$, $\mathcal{T}$ may have been modified).
					
					\If{$|A| \ge m+1$}
					\State \Return \reject
					\EndIf
					\EndFor
					\State \Return \accept
				\end{algorithmic}
			\end{algorithm}
			
		\end{subalgorithms}
		
		\begin{theorem} \label{th:alg-m-supp-adaptive-correct}
			Algorithm \ref{alg:m-supp-adaptive} is a one-sided $\eps$-test of $\mathcal{S}_m$ using $O(\eps^{-1} m \log m \cdot \min\{\log \eps^{-1}, \log m\})$ many queries.
		\end{theorem}
		\begin{proof}
			If $\eps \ge \frac{1}{m^2}$ then the correctness is implied by Theorem \ref{th:alg-m-supp-na-correct}, and the query complexity is $O(\eps^{-1} \log \eps^{-1} \cdot m \log m)$. Since $\eps \ge \frac{1}{m^2}$, $\log \eps^{-1} \le 2 \log m$ and the query complexity is bounded by $O(\eps^{-1} m \log m \cdot \min\{\log \eps^{-1}, \log m\})$.
			
			If $\eps < \frac{1}{m^2}$, the query complexity of the first phase is $O(m^3 \log m + \eps^{-1} m \log^2 m)$ (Lemma \ref{lemma:alg-m-supp-adaptive-first-phase-correct}). The query complexity of constructing $\mathcal{T}$ for the first time (between the phases) is $m^2$ queries, which is at most $\eps^{-1}$ since $\eps < \frac{1}{m^2}$. The query complexity of the second phase is $O(\eps^{-1}) \cdot 3m = O(\eps^{-1} m)$ (Lemma \ref{lemma:m-supp-adaptive-second-phase-iteration-correct}). Overall, the query complexity of the algorithm is $O(m^3 \log m + \eps^{-1} m \log^2 m)$. Since $m^2 <\eps^{-1}$, $m^3 < \eps^{-1} m$ and thus the query complexity is bounded by $O(\eps^{-1} m \log^2 m)$. Since $\log m \le \log \eps^{-1}$, it is bounded by $O(\eps^{-1} m \log m \cdot \min\{\log \eps^{-1}, \log m\})$ as well.
			
			Perfect completeness is trivial, since the algorithm rejects only if $|A| \ge m+1$, where $|A|$ is a set of fully distinguishable samples.
			
			For soundness, consider an input distribution $P$ that is $\eps$-far from $\mathcal{S}_m$.
			
			By Lemma \ref{lemma:alg-m-supp-adaptive-first-phase-correct}, with probability $\frac{3}{4}$, either $|A| \ge m+1$ or $\Ct\left[d(x,A) | d(x,A) > \frac{1}{2m}\right] \le \frac{1}{2}\eps$. In the first case the algorithm rejects immediately. Otherwise, we analyze the second phase.
			
			Assume that the second phase of the algorithm had infinitely many iterations. By Lemma \ref{lemma:m-supp-adaptive-second-phase-iteration-correct}, as long as $|A| \le m$, every iteration of the second phase extends $A$ with probability at least $\frac{1}{4}m \eps$. The number of iterations until $A$ has $m+1$ elements is a sum of at most $m$ geometric random variables, each having success probability at least $\frac{1}{4}m\eps$. The expected number of iterations until $A$ has $m+1$ elements is thus bounded by $m\cdot 4\eps^{-1}/m = 4\eps^{-1}$. By Markov's inequality, with probability at least $\frac{11}{12}$, this number of iterations is at most $48\eps^{-1}$.
			
			To conclude: with probability at least $\frac{3}{4}$, there exists a distinguishable set $A$ for which $|A| \ge m+1$ or $\Ct[d(x,A) | d(x,A) > \frac{1}{2m}] \le \frac{1}{2}\eps$. In the first case the algorithm rejects the input immediately, and in the second case, with probability at least $\frac{11}{12}$, the algorithm rejects the input within the first $\ceil{48\eps^{-1}}$ samples of the second phase. Overall, the probability to reject an $\eps$-far input is at least $\frac{3}{4} \cdot \frac{11}{12} > \frac{2}{3}$.
		\end{proof}
		
		\section{Superlinear lower-bound for one-sided adaptive $m$-support test}
		\label{sec:lbnd-adaptive-one-sided}
		
		\subsection{Lower bound on the size of witnesses against $\mathcal{S}_m$}
		
		We show that every witness against $\mathcal{S}_m$ must be of size at least $\Omega(m \log m)$, hence every one-sided $\eps$-test algorithm for $\mathcal{S}_m$ must use $\Omega(m \log m)$ queries as well. We use this to show an $\Omega(\eps^{-1} m \log m)$ lower bound for non-adaptive algorithms, and after a short discussion we extend this result to adaptive algorithms as well.
		
		\begin{definition}[Capacity of an edge cover]
			Let $G$ be a graph over a set $V$ vertices and let $\mathcal{G} = (G_1,\ldots,G_k)$ be a sequence of graphs over $V_1,\ldots,V_k \subseteq V$ such that $G = \bigcup_{i=1}^k G_i$. We define the \emph{capacity} of $\mathcal{G}$ as $\mathrm{cap}(\mathcal{G}) = \sum_{i=1}^k |V_k|$.
		\end{definition}
		
		The following observation follows directly from the definition of capacity.
		\begin{observation} \label{obs:capacity-queries-equivalence}
			Let $P$ be a distribution over $\{0,1\}^n$, $x_1,\ldots,x_s \in \supp(P)$ be a set of samples and $Q \subseteq \{1,\ldots,s\} \times \{1,\ldots,n\}$ be a query set. Let $S_1,\ldots,S_n$ be the index-specific query sets, that is, $Q = \bigcup_{j=1}^n (S_j \times \{j\})$. In other words, for every $j$, all samples in $S_j$ are queried at the index $j$.
			Let $\mathcal{G} = (G_1,\ldots,G_n)$ be the edge cover of the contradiction graph (Definition~\ref{def:contradiction-graph}) implied by $(x_1,\ldots,x_s;Q)$: for every $1 \le j \le n$, $G_j$ is the complete bipartite graph whose vertices are $S_j$ and the sides are:
			\begin{align*}
				& L_j = \{i \in S_j | (x_i)_j = 0 \}, && R_j = \{i \in S_j | (x_i)_j = 1 \}
			\end{align*}
			
			In this setting, $\mathrm{cap}(\mathcal{G}) = |Q|$.
		\end{observation}
		
		\begin{lemma}[\cite{hansel64, ks67,alon2023bipartite}] \label{lemma:hansel}
			Let $V$ be a set of vertices, and let $\mathcal{G} = (G_1,\ldots,G_k)$ be an edge cover of the $V$-clique such that all graphs $G_1,\ldots,G_k$ are bipartite. Then $\mathrm{cap}(\mathcal{G}) \ge |V| \log_2 |V|$.
		\end{lemma}
		
		\begin{lemma} \label{lemma:clique-witness-reduction}
			Let $G$ be a graph over a set $V$ of vertices that is not $m$-colorable, and let $\mathcal{G} = (G_1,\ldots,G_k)$ be an edge cover of $G$ such that all graphs $G_1,\ldots,G_k$ are bipartite. Then $\mathrm{cap}(\mathcal{G}) \ge (m+1) \log_2 (m+1)$.
		\end{lemma}
		\begin{proof}
			Without loss of generality we assume that $G$ is exactly $m+1$-colorable (otherwise we can just omit vertices one by one until it is). Let $U_1,\ldots,U_{m+1}$ be a coloring of $G$ using $m+1$ colors (that is, $U_i$ is an independent set for every $1 \le i \le m+1$ and $\bigcup_{i=1}^k U_i = V$).
			
			Let $\hat{G}$ be a graph over $\{1,\ldots,m+1\}$ such that the edge $\{i,j\}$ exists if and only if there is an edge between a vertex in $U_i$ and a vertex in $U_j$. We define the edge cover $\hat{\mathcal{G}} = (\hat{G}_1,\ldots,\hat{G}_k)$ similarly: the vertex $i$ belongs to $\hat{G}_j$ if some vertex in $U_i$ belongs to $G_j$. Note that the sides are implied since all vertices of $U_i$ must be on the same side. Note that $\mathrm{cap}(\hat{\mathcal{G}}) \le \mathrm{cap}(\mathcal{G})$, since every vertex in $\hat{\mathcal{G}}$ represents a (disjoint) set of vertices in $\mathcal{G}$.
			
			Observe that $\hat{G}$ must be a clique: if an edge $\{i,j\}$ is missing, then $U_i \cup U_j$ is an independent set and hence $G$ is $m$-colorable. Hence $\hat{\mathcal{G}}$ is an edge cover of $m+1$-clique using bipartite graphs, and by Lemma \ref{lemma:hansel} its capacity must be at least $(m+1) \log_2 (m+1)$. Hence $\mathrm{cap}(\mathcal{G}) \ge \mathrm{cap}(\hat{\mathcal{G}}) \ge (m+1) \log_2 (m+1)$.
		\end{proof}

		\begin{proposition} \label{prop:lbnd-one-sided-witness-m-log-m}
			Every witness against belonging to $\mathcal{S}_m$ must be at least $m \log_2 m$-bits long. In particular, every one-sided $\eps$-testing algorithm for $\mathcal{S}_m$ must make at least $m \log_2 m$ queries.
		\end{proposition}
		
		\begin{proof}
			By Lemma \ref{lemma:witness-iff-not-colorable}, there exists a witness against $\mathcal{S}_m$ if and only if the contradiction graph is not $m$-colorable. By Lemma \ref{lemma:clique-witness-reduction}, the capacity of every bipartite cover of the contradiction graph is at least $(m+1) \log_2 (m+1)$. By Observation \ref{obs:capacity-queries-equivalence}, the number of queries must be at least $(m+1) \log_2 (m+1)$ as well.
		\end{proof}
		
		\subsection{An improved bound for non-adaptive algorithms}
		
		The lower bound on the size of a witness implies a trivial bound of $\Omega(m \log m)$ queries for every one-sided $\eps$-test of $\mathcal{S}_m$ for any $0<\eps<1$. To extend this result to a slightly better $\Omega(\eps^{-1} m \log m)$ bound for non-adaptive algorithms, we first need the following almost-trivial observation.
		\begin{observation}\label{obs:graph-coloring-in-parts}
			Let $G$ be a graph over $V = V_0 \cup V_1$ (where $V_0 \cap V_1 = \emptyset$). If the subgraph of $G$ induced by $V_1$ is $k$-colorable, then $G$ is $(|V_0|+k)$-colorable.
		\end{observation}
		
		Based on this observation we can improve the $\Omega(m \log m)$ bound for non-adaptive algorithms.
		\begin{proposition} \label{prop:lbnd-one-sided-m-locally-bounded}
			Every one-sided non-adaptive $\eps$-testing algorithm for $\mathcal{S}_m$ must make at least $\Omega(\eps^{-1} m \log m)$ many queries.
		\end{proposition}
		\begin{proof}
			Consider the following $D_\mathrm{no}$ distribution: first, we choose $a_1,\ldots,a_{2m} \sim \{0,1\}^n$ uniformly and independently, and then we return
			\begin{align*}
				P \sim \begin{cases}
					0 & \mathrm{with\ probability\ } 1 - 10\eps \\
					a_i & \mathrm{with\ probability\ } \frac{5\eps}{m}, 1 \le i \le 2m
				\end{cases}
			\end{align*}
			With probability $1 - o(1)$, all $a_i$s are $0.49$-far from each other and from the zero vector, hence the distance of $P$ from $\mathcal{S}_m$ is at least $0.24 \cdot m \cdot \frac{5\eps}{m} > \eps$.
			
			Let $Q$ be a query set with $|Q| < \frac{1}{100}\eps^{-1} m \log_2 m$ queries over $s$ samples. For every $1 \le i \le s$, let $q_i$ be the number of queries in the $i$-th sample. Note that $q = \sum_{i=1}^s q_i$, by definition. This is where we use the assumption that the algorithm is non-adaptive, since $q_1,\ldots,q_s$ may not be pre-determined for adaptive algorithms.
			
			By linearity of expectation, the expected number of queries applied to non-zero samples is bounded by $10\eps\sum_{i=1}^s q_i = 10\eps \cdot |Q| < \frac{1}{10} m \log_2 m$. With probability higher than $\frac{4}{5}$, this number of queries is smaller than $\frac{1}{2} m \log_2 m$, and in this case, there cannot be a witness for $m$ distinct elements among $a_1,\ldots,a_{2m}$. This means that, with the same probability, there is no witness for $m+1$ elements among $0, a_1,\ldots, a_{2m}$. The lower bound follows from Lemma \ref{lemma:yao-lbnd-one-sided}.
		\end{proof}
		
		\subsection{Extending the bound to adaptive algorithms}
		
		Proposition \ref{prop:lbnd-one-sided-m-locally-bounded} only applies to non-adaptive algorithms, and may also apply to the class of locally-bounded adaptive algorithms defined in \cite{adar23} (but we do not show it here). The bottleneck of the proof is the need of having a ``good'' upper bound on the maximum number of queries per sample as compared to the expected number of queries per sample, which is impossible for adaptive algorithms. To extend the proof, we must make the algorithm completely clueless unless it queries every individual sample within a fixed portion of the maximum. To achieve these goals we use two concepts: very short strings (to reduce the maximum number of queries per sample -- but the extension for arbitrarily long strings is then proved as a simple corollary) and secret sharing (see below).
		
		\begin{lemma}[Definition and existence of secret-sharing code ensembles, \cite{beflr2020}] \label{lemma:shared-secret-exists}
			There exist some constants $\delta,\zeta > 0$ and $\eta > 1$, for which: for every $k \ge 1$, there exist $m(k) \le 2k$ divisible by $3$, and some $n(k)$ with $n(k) \le \eta \log_2 m(k)$ for which there exists a code-ensemble $\mathcal{H} : \{0,1\} \to \mathcal{P}(\{0,1\}^{n(k)})$ with the following properties:
			
			\begin{itemize}
				\item Sufficiently large: $|\mathcal{H}(0)| = |\mathcal{H}(1)| = \frac{2}{3} m(k)$.
				\item Fixed lower bound on the distance: for every $u,v \in \mathcal{H}(0) \cup \mathcal{H}(1)$ either $u=v$ or $d(u,v) > \delta$.
				\item Shared secret: for every set $I \subseteq \{1,\ldots,n(k)\}$ of size $|I| \le \zeta n(k)$, and for every $w \in \{0,1\}^{|I|}$,
				\[\Pr_{u \sim \mathcal{H}(0)}[u|_I = w] = \Pr_{u \sim \mathcal{H}(1)}[u|_I = w]\]
			\end{itemize}
		\end{lemma}
		
		By the specific construction of \cite{beflr2020}, we can have $\delta=\frac{1}{30}$, $\zeta=\frac{1}{12}$ and $\eta = 4$, and also that every restriction of $\mathcal{H}(0)$ (or $\mathcal{H}(1)$) to a set $I$ of at most $\zeta n(k)$ bits is uniform over $\{0,1\}^I$. Note that the following does not rely on the exact values of $\delta, \zeta, \eta$, only on the constraints guaranteed by the lemma.
		
		\begin{theorem} \label{th:lbnd-one-sided-m-adaptive-almost}
			Consider $\delta$, $\zeta$, $\eta$ of Lemma \ref{lemma:shared-secret-exists}. For every $\eps < \frac{1}{4}\delta$ and every $k \ge 1$ there exist $m(k) \ge 2k$, $n(k) \le \eta \log_2 m(k)$ and a distribution $P$ over $\{0,1\}^{n(k)}$ that is $\eps$-far from $\mathcal{S}_{m(k)}$ for which every $\eps$-testing adaptive algorithm that makes less than $\frac{1}{96}\zeta\delta\eps^{-1}m(k) \log_2 m(k)$ queries cannot find a witness against $P \in \mathcal{S}_{m(k)}$ with probability $\frac{1}{3}$.
		\end{theorem}
		\begin{proof}
			Let $k \ge 1$ and let $m = m(k) \ge 2k$ and $\log_2 m \le n(k) \le \eta \log_2 m$ as per Lemma \ref{lemma:shared-secret-exists}. For every $0 < \eps < \frac{1}{4}\delta$, we define the following distribution:
			\begin{align*}
				P : \begin{cases}
					a & \mathrm{with\ probability\ } \frac{1 - 4\delta^{-1}\eps}{2m/3},\; \forall a \in \mathcal{H}(0) \\
					b & \mathrm{with\ probability\ } \frac{4\delta^{-1}\eps}{2m/3}, \;\forall b \in \mathcal{H}(1)
				\end{cases}
			\end{align*}
			
			By the secret sharing property, for every set $Q$ of at most $\zeta n(k)$ indices and for every $w \in \{0,1\}^{|Q|}$,
			\[\Pr_{x\sim P}\left[x \in \mathcal{H}(1) \cond x|_Q = w \right] = \Pr_{x\sim P}\left[x \in \mathcal{H}(1)\right] = 4\delta^{-1}\eps \]
			
			Observe that $P$ is supported by exactly $\frac{4}{3}m$ elements. For every set $A$ of $m$ elements, at least $\frac{1}{3}m$ elements in the support of $P$ are $\frac{1}{2}\delta$-far from every element in $A$. The minimum probability mass of an individual element in $P$ is at least $\frac{4\delta^{-1}\eps}{2m/3}$, hence
			\[d(P,\mathcal{S}_A) > \frac{1}{3}m \cdot \frac{4\delta^{-1}\eps}{2m/3} \cdot \frac{1}{2}\delta = \eps \]
			This holds for every $A$ of size $m$, hence $d(P,\mathcal{S}_m) > \eps$.
			
			Assume that $m \ge \max\left\{2^{2\zeta^{-1}},6\right\}$. Let $n = n(k)$ be the length of the string and $n' = \floor{\zeta n(k)}$ be the number of bits whose reading supplies no information about $x \sim P$ coming from $\mathcal{H}(0)$ or $\mathcal{H}(1)$. Consider a deterministic adaptive algorithm $T$ with $q < \frac{1}{96} \zeta \delta \eps^{-1} m \log_2 m$ queries. The number of queries per sample is trivially bounded by $n = O(\log_2 m)$, since this is the bit length of a sample. 
			
			We color every query in $T$ with either red or green. In every path in the tree, for every $1 \le i \le s$, the first $n'$ queries of the $i$-th sample are red and the others are green. Note that the color of a query is unambiguous even if it is common to multiple paths, since the coloring only takes the path from the root into account. Note that in every path in $T$, the number of queries is bounded by $n/n'$-times the number of red queries, which is:
			\[\frac{n}{n'} = \frac{n}{\floor{\zeta n}} \le \frac{n}{\zeta n - 1} \le \frac{\log_2 m}{\zeta \log_2 m - 1} \le \frac{2\zeta^{-1}}{\zeta\cdot 2\zeta^{-1} - 1} = 2\zeta^{-1}\]
			
			We would like to bound the expected number of queries applied to $\mathcal{H}(1)$-samples to apply Proposition \ref{prop:lbnd-one-sided-witness-m-log-m}. For every $1 \le i \le q$, let $R_i$ be an indicator for ``the $i$-th query is red'' and $X_i$ be an indicator for ``the $i$-th query is red, and it applies to $\mathcal{H}(1)$-sample''. Let $X = \sum_{i=1}^q X_i$ be the number of red queries applied to $\mathcal{H}(1)$-samples and $Y$ be the total number of queries (of any color) applied to $\mathcal{H}(1)$-samples.
			
			By the secret sharing property, and the definition of ``red queries'' as the first $\floor{\zeta n}$ queries, $\E[X_i |\; R_i = 1] = 4\delta^{-1}\eps$, since the restriction of a sample drawn from $P$ to at most $\zeta n$ indices distributes exactly the same regardless of whether it belongs to $\mathcal{H}(0)$ or to $\mathcal{H}(1)$. Hence
			\[\E[X_i] = \Pr[R_i = 0]\underbrace{\E[X_i \mid R_i = 0]}_{=0} + \underbrace{\Pr[R_i = 1]}_{\le 1}\underbrace{\E[X_i \mid R_i = 1]}_{=4\delta^{-1}\eps} \le 4\delta^{-1}\eps\]
			By linearity of expectation, $\E[X] = \sum_{i=1}^q \E[X_i] \le 4\delta^{-1}\eps q < \frac{1}{24} \zeta m \log_2 m$. Recall that $Y \le 2\zeta^{-1} X$ for every sufficiently large $m$, hence $\E[Y] \le \frac{1}{12} m \log_2 m$. By Markov's inequality, $\Pr[Y > \frac{1}{8} m \log_2 m] < \frac{2}{3}$.
			
			Consider a witness against $P \in \mathcal{S}_m$. By Observation \ref{obs:graph-coloring-in-parts}, since $\supp(P) = \mathcal{H}(0) \cup \mathcal{H}(1)$ and $|\mathcal{H}(0)| = \frac{2}{3}m$, there must exist a ``sub-witness'' against $|\mathcal{H}(1)| \le \frac{1}{3}m$, meaning that the restriction of the queries and answers to $\mathcal{H}(1)$-samples forms a witness against having support at most $\frac{1}{3}m$. By Proposition \ref{prop:lbnd-one-sided-witness-m-log-m}, this sub-witness must consist of at least $\frac{1}{3}m \log_2 (\frac{1}{3}m)$ bits of $\mathcal{H}(1)$-samples. For $m \ge 6$, this bound requires more than $\frac{1}{8}m \log_2 m$ bits.
			
			With probability greater than $\frac{1}{3}$, there are fewer than $\frac{1}{8} m \log_2 m$ queries in $\mathcal{H}(1)$-samples, hence they cannot form a witness against $|\mathcal{H}(1)| \le \frac{1}{3}m$. In particular, in this case, there is no witness against $P \in \mathcal{S}_m$. For every $k \ge 1$ there exists $m \ge 2k$ for which this result applies for every $\eps > 0$ smaller than some global constant. The lower bound follows from Lemma \ref{lemma:yao-lbnd-one-sided}.
		\end{proof}
		
		Theorem \ref{th:lbnd-one-sided-m-adaptive-almost} shows a lower bound of $\Omega(\eps^{-1} m \log m)$ queries for every one-sided adaptive $\eps$-testing of $\mathcal{S}_m$ for some fixed $n$ (that depends on $m$). However, a ``full'' property testing lower bound would need to apply for at least an infinite set of value for $n$. The following corollary uses a simple repeating technique to provide this.
		
		\begin{lemma} \label{lemma:gemini}
			Let $P$ be a distribution over $\{0,1\}^n$. Let $f_{\ell} : \{0,1\}^n \to \{0,1\}^{\ell n}$ be defined by the $\ell$-fold repetition, that is $f_{\ell}(x_1,\ldots,x_n)=y$ where $y_{i\cdot n+j}=x_j$ for $1\leq j\leq n$ and $0\leq i\leq\ell-1$, and consider the distribution $\hat{P} = f_\ell(P)$. Then for every $m \ge 1$, $d(P, \mathcal{S}_m) \le d(\hat{P}, \mathcal{S}_m)$.
		\end{lemma}
		\begin{proof}
			Let $\hat{A} = \{\hat{x}_1,\ldots,\hat{x}_m\}$ be a set that realizes the distance $d(\hat{P}, \mathcal{S}_m)$. Let $g : \supp(\hat{P}) \to \hat{A}$ be the mapping provided by Lemma \ref{lemma:realizing-the-distance-from-a-set}. For every $0 \le i \le \ell - 1$, let $h_i : \{0,1\}^{\ell n} \to \{0,1\}^n$ be the mapping $h_i(x) = x_{\{i\cdot n + 1,\ldots,i\cdot n + n\}}$. Choose $0 \le i \le \ell - 1$ uniformly, and let $h = h_i$. Let $A = \{h(\hat{x}_1),\ldots,h(\hat{x}_m)\}$.
			
			\begin{eqnarray*}
				d(P,\mathcal{S}_m)
				&\le& \E_{i\sim\{1,\ldots,n\}} [d(P,\mathcal{S}_A)]
				= \E_{x\sim P,i\sim\{1,\ldots,n\}}[d(x,A)]
				= \sum_{x \in \supp(P)} \Pr_P[x] \E_i[d(x,\mathcal{S}_A)] \\
				&\underset{(*)}\le& \sum_{x \in \supp(P)} \Pr_P[x] \E_i[d(x,h(g(f_\ell(x))))] \\
				& \underset{(\dag)}=& \sum_{x \in \supp(P)} \Pr_P[x] d(f_\ell(x),g(f_\ell(x))) \\
				&=& \sum_{x \in \supp(P)} \Pr_{\hat{P}}[f_\ell(x)] d(f_\ell(x),g(f_\ell(x))) \\
				&=& \sum_{\hat{x} \in \supp(\hat{P})} \Pr_{\hat{P}}[\hat{x}] d(\hat{x},g(\hat{x}))
				= d(\hat{P}, \mathcal{S}_{\hat{A}}) = d(\hat{P}, \mathcal{S}_m)
			\end{eqnarray*}
			
			The starred transition is correct since $h(g(f_\ell(x))) \in A$ by its definition. The daggered transition is correct since:
			\begin{eqnarray*}
				\E_{i\sim\{1,\ldots,n\}}[d(x,h(g(f_\ell(x))))]
				&=& \frac{1}{\ell} \sum_{i=0}^{\ell-1} d(x, h_i(f(x))) \\
				&=& d(x \cdots x, h_0(g(f_\ell(x))) \cdots h_{\ell-1}(g(f_\ell(x))))
				= d(f_\ell(x), g(f_\ell(x)))
			\end{eqnarray*}
		\end{proof}
		
		\begin{corollary} \label{col:lbnd-one-sided-m-adaptive}
			Every one-sided (possibly adaptive) $\eps$-testing algorithm for $\mathcal{S}_m$ must make at least $\Omega(\eps^{-1} m \log m)$ many queries, for infinitely many values of the string length $n$.
		\end{corollary}
		\begin{proof}
			For a proper choice of $m \ge 2$, $q \ge 1$ and $\eps > 0$, let $P$ be a distribution over $\{0,1\}^n$ (for some $n$ that may depend on $m$ and $\eps$) that is $\eps$-far from $\mathcal{S}_m$ for which, for every deterministic adaptive algorithm $T$ that makes at most $q$ queries, the probability that it finds a witness against $P \in \mathcal{S}_m$ is less than $\frac{1}{3}$. For every $\hat{n} \ge n$, let $\hat{P} = f_{\ceil{\hat{n} / n}}(P)$ (as per Lemma \ref{lemma:gemini}). Note that $P$ is $\eps$-far from $\mathcal{S}_m$ as well.
			
			Every bit of a sample drawn from $\hat{P}$ corresponds to a concrete bit of a sample drawn from $P$, hence every witness against $\hat{P} \in \mathcal{S}_m$ is also a witness against $P \in \mathcal{S}_m$. The probability to find this witness using less than $q$ queries is less than $\frac{1}{3}$, hence every one-sided adaptive $\eps$-testing algorithm for $\mathcal{S}_m$ must use at least $q$ queries, even for distributions over arbitrarily long strings. This completes the proof since $q = \Omega(\eps^{-1} m \log m)$.
		\end{proof}
		
		As a final remark, note that the corollary actually holds for sets of values of $m$ and $n$ that are not very sparse -- all large enough integers of the form $3\cdot 2^{4k}$ for $m$, and a set of positive density (that depends on $m$) for $n$.
		
		\section*{Acknowledgement}
		We thank Noga Alon for pointing out the results of \cite{hansel64} and \cite{ks67} that we use in Section~\ref{sec:lbnd-adaptive-one-sided}.
		
		\bibliographystyle{alpha}
		\bibliography{main.bib}
		
		\appendix
		
		\section{Probabilistic Bounds}
		
		We prove and recall here some technical probabilistic bounds that are used in our proofs.
		
		\begin{lemma}\label{lemma:mapping-contribution}
			Let $X$ be a non-negative random variable and let $f,g : \mathbb R \to \mathbb R$ be two non-negative functions. Assume that $\Pr\left[g(X) \ge f(X) \cond B\right] = 1$ for some event $B$. Then $\E\left[g(X) \cond B\right] \ge \Ct\left[f(X) \cond B\right]$.
		\end{lemma}
		\begin{proof}
			$\E\left[g(X) \cond B\right] \ge \Ct\left[g(X) \cond B\right] = \Pr[B]\E\left[g(X) \cond B\right] \ge \Pr\left[B\right]\E\left[f(X) \cond B\right] = \Ct\left[f(X) \cond B\right]$
		\end{proof}

		\begin{lemma}[A technical bound] \label{lemma:pr-chernoff-zero}
			Let $X$ be a sum of independent variables $X_1,\ldots,X_n$, where each evaluates to $1$ with probability $p_i$ and evaluates to $0$ otherwise. Then $\Pr[X = 0] \le e^{-\E[X]}$.
		\end{lemma}
		\begin{proof}
			$\Pr[X = 0] = \prod_{i=1}^n (1 - p_i) \le e^{-\sum_{i=1}^n p_i} = e^{-\E[X]}$
		\end{proof}
		
		\begin{observation}\label{obs:exp-linearization}
			For $n \ge 1$ and $0 \le p < \frac{1}{2n}$, $1 - (1 - p)^n \ge \frac{1}{2}np$.
		\end{observation}
		\begin{proof} Since $np<1/2$, we have 
			$(1-p)^n\le 1-np+(np)^2\le 1-np+\frac{1}{2}np=1-\frac{1}{2}np$.
		\end{proof}
		
		\begin{lemma}[Multiplicative Chernoff's Bound]
			\label{lemma:chernoff-mult}
			Let $X_1,\ldots,X_m$ be independent variables in $\{0,1\}$. Let $X = \sum_{i=1}^m X_i$, then for every $\delta > 0$:
			\begin{eqnarray*}
				\Pr[X < (1-\delta) \E[X]] &<& \left(\frac{e^{-\delta}}{(1 - \delta)^{1 - \delta}}\right)^{\E[X]} \\
				\Pr[X > (1 + \delta)\E[X]] &<& \left(\frac{e^\delta}{(1 + \delta)^{1 + \delta}}\right)^{\E[X]}
			\end{eqnarray*}
		\end{lemma}
		
		\begin{lemma}[Multiplicative Chernoff for well-dependent variables with a goal] \label{lemma:chernoff-with-goal}
			Let $\mathcal{G}\subset \mathbb{R}^*$ be a set of goal sequences, satisfying that if $u$ is a prefix of $v$ and $u\in\mathcal{G}$ then $v\in\mathcal{G}$. Additionally let $R_1,\ldots,R_m$ be a set of random variables and $p_1,\ldots,p_m$ be values in $[0,1]$, such that for every $1\leq i\leq m$ and $v=(r_1,\ldots,r_{i-1})\in\mathbb{R}^{i-1}\setminus\mathcal{G}$ (that can happen with positive probability) we have $\Pr\left[R_i\ne 0 \cond R_1=r_1,\ldots,R_{i-1}=r_{i-1}\right] \ge p_i$. For every $1 \le i \le m$, let $X_i \in \{0,1\}$ be an indicator for $R_i \ne 0$ and $X = \sum_{i=1}^m X_i$. Under these premises, for every $0 < \delta < 1$,
			\[\Pr\left[((R_1,\ldots,R_m)\notin\mathcal{G}) \wedge \left(X < (1-\delta) \sum_{i=1}^m p_i\right)\right]
			< \left(\frac{e^{-\delta}}{(1 - \delta)^{1 - \delta}} \right)^{\sum_{i=1}^m p_i}\]
		\end{lemma}
		\begin{proof}
			We first define auxiliary random variables $Y_1,\ldots,Y_m \in \{0,1\}$ that will depend on $R_1,\ldots,R_m$ but will be independent of each other. To draw the value of $Y_i$, we consider $R_1,\ldots,R_i$. Considering their respective values $r_1,\ldots,r_i$, if $(r_1,\ldots,r_{i-1})\in\mathcal{G}$ we take $Y_i$ to be equal to $1$ with probability $p_i$ and to $0$ with probability $1-p_i$, independently of all other choices so far. If $(r_1,\ldots,r_{i-1})\notin\mathcal{G}$, then we set $\alpha_i=\Pr[R_i\neq 0|R_1=r_1,\ldots,R_{i-1}=r_{i-1}]$ and choose $Y_i$ according to $r_i$: If $r_i=0$ (meaning in particular that $X_i=0$) then we choose $Y_i=0$. If $r_i\neq 0$, we choose $Y_i=1$ with probability $\frac{p_i}{\alpha_i}$ and choose $Y_i=0$ with probability $\frac{\alpha_i-p_i}{\alpha_i}$. This last choice is drawn independently of previous choices (note that in particular these are indeed probabilities between $0$ and $1$, since by the assumptions of the lemma $p_i\leq\alpha_i\leq 1$). We also define the sum $Y=\sum_{i=1}^mY_i$.
			
			It is not hard to see that $\Pr[Y_i=1]=p_i$ for every $1\leq i\leq m$. To conclude the proof, it remains to show that $Y_1,\ldots,Y_m$ are indeed independent (when not conditioning on the other random variables defined over our probability space), and that it is always the case that $Y\leq X$ or $(R_1,\ldots,R_m)\in\mathcal{G}$ (or both), since then we can use the multiplicative Chernoff bound to conclude that
			\begin{eqnarray*}
				&& \Pr\left[((R_1,\ldots,R_m) \notin \mathcal{G}) \wedge \left(X < (1-\delta) \sum_{i=1}^m p_i \right) \right] \\
				&\le& \Pr\left[((R_1,\ldots,R_m) \notin \mathcal{G}) \wedge \left(Y < (1-\delta) \sum_{i=1}^m p_i \right) \right] \\
				&\le& \Pr\left[Y < (1-\delta)\E[Y]\right]
				< \left(\frac{e^{-\delta}}{(1 - \delta)^{1 - \delta}}\right)^{\sum_{i=1}^m p_i}
			\end{eqnarray*}
			
			For the independence assertion, we need to show that for every sequence of values $(b_1,\ldots,b_{i-1})\in\{0,1\}^{i-1}$ we have $\Pr[Y_i=1|Y_1=b_1,\ldots,Y_{i-1}=b_{i-1}]=p_i$. We note that it is enough to show that for every sequence $(r_1,\ldots,r_{i-1})\in\supp(R_1,\ldots,R_{i-1})$ we have $\Pr[Y_i=1|R_1=r_1,\ldots,R_{i-1}=r_{i-1}]=p_i$, since the choices of $Y_1,\ldots,Y_{i-1}$ depend only on the values of $R_1,\ldots,R_{i-1}$ (and possible additional independent coin tosses). To show the latter, we go over the cases. If $(r_1,\ldots,r_{i-1})\in\mathcal{G}$ then $Y_i$ was explicitly defined to be $1$ with probability exactly $p_i$. If $(r_1,\ldots,r_{i-1})\notin\mathcal{G}$, we write 
			\begin{eqnarray*}
				&& \Pr\left[Y_i = 1 \cond R_1=r_1,\ldots,R_{i-1}=r_{i-1}\right] \\
				&=& \Pr\left[Y_i=X_i=1 \cond R_1=r_1,\ldots,R_{i-1}=r_{i-1}\right] \\
				&=& \Pr\left[Y_i=1 \cond X_i=1,R_1=r_1,\ldots,R_{i-1}=r_{i-1}\right] \cdot \Pr\left[X_i=1 \cond R_1=r_1,\ldots,R_{i-1}=r_{i-1}\right] \\
				&=& \frac{p_i}{\alpha_i}\cdot\alpha_i = p_i
			\end{eqnarray*}
			
			To show the conditional inequality assertion, note that $(R_1,\ldots,R_m)\notin\mathcal{G}$ implies in particular that $(R_1,\ldots,R_{i-1})\notin\mathcal{G}$ for every $1\leq i\leq m$. Hence, all the choices of $Y_i$ in this case are made so that $Y_i\leq X_i$ for $1\leq i\leq m$, and in particular $Y\leq X$.
		\end{proof}

		\section{Proof of Proposition \ref{prop:lbl-inv-and-one-side-is-supp}}
		\label{apx:proof-of-proposition}
		
		To prove Proposition 3.26, we need the following lemma.
		
		\begin{lemma}
			\label{lemma:one-sided-is-closed-under-conditioning}
			Consider a property $\mathcal{P}$ of distributions over $\{0,1\}^n$ that has a one-sided $\eps$-test for every $\eps > 0$, and consider some $P\in\mathcal{P}$. For every distribution $Q$ for which $\supp(Q) \subseteq \supp(P)$, $Q \in \mathcal{P}$ as well.
		\end{lemma}
		\begin{proof}
			Let $\mathcal{A}_{\eps,n}$ be an $\eps$-testing algorithm for $\mathcal{P}$ in the Huge Object model that draws $s(\eps,n)$ samples and makes $q(\eps,n)$ queries. For every $\eps$, if $\mathcal{A}_{\eps,n}$ rejects $Q$ with positive probability, then there must be a sequence of samples $X_1,\ldots,X_{s(\eps,n)}$ and a sequence of random query choices $Y_1,\ldots,Y_{q(\eps,n)}$ for which the algorithm rejects. If we run $\mathcal{A}_{\eps,n}$ on $P$ as its input, there is a positive probability to draw exactly the same sequence $X_1,\ldots,X_{s(\eps,n)}$ (since $\supp(Q) \subseteq \supp(P)$), and to make the exact same random query choices $Y_1,\ldots,Y_{q(\eps,n)}$. In this case, $\mathcal{A}_{\eps,n}$ rejects $P$, a contradiction to its one-sideness. Hence $\mathcal{A}_{\eps,n}$ must always accept $Q$ and this holds for every $\eps > 0$, which implies $Q \in \mathcal{P}$.
		\end{proof}
		
		We now recall the proposition to be proved.
		
		\propZlblZinvZandZoneZsideZisZsupp*
		
		\begin{proof}
			Let $f(n)$ be as follows:
			\begin{align*}
				f(n) = \begin{cases}
					0 & \mathcal{P} \cap \mathcal{D}(\{0,1\}^n) = \emptyset \\
					\max_{P \in \mathcal{P}} |\supp(P)| & \mathrm{otherwise}
				\end{cases}
			\end{align*}
			
			In the first case, the property is empty for $n$-bit strings. In the second case, a maximum must exist, and it cannot be more than $2^n$. For every $n$ for which $f(n) > 0$, we also define $P_n$ as one of the distributions that demonstrate the maximum. By the definition of $f(n)$, $\mathcal{P}$ does not contain any distribution that is supported by more than $f(n)$ elements.
			
			Consider some $n$ for which $f(n) > 0$, and consider some distribution $P$ over $\{0,1\}^n$ that is supported by at most $f(n)$ elements. Let $\sigma : \{0,1\}^n \to \{0,1\}^n$ be a permutation for which $\supp(P) \subseteq \supp(\sigma(P_n))$. By the label invariance of $\mathcal{P}$, $\sigma(P_n) \in \mathcal{P}$. By Lemma \ref{lemma:one-sided-is-closed-under-conditioning}, $P \in \mathcal{P}$, because its support is a subset of the support of $\sigma(P_n)$ which belongs to $\mathcal{P}$. Hence $\mathcal{P}$ contains all distributions that are supported by at most $f(n)$ elements, as desired.
		\end{proof}
		
	\end{document}